\newcommand{\CLASSINPUTbaselinestretch}{1.2}
\documentclass[12pt,journal,onecolumn]{IEEEtran}

\usepackage{geometry}
\usepackage[english]{babel}
\usepackage{enumitem}
\usepackage[latin1]{inputenc}
\usepackage{boondox-cal}
\usepackage{enumerate}
\usepackage{color}
\usepackage[T1]{fontenc}
\usepackage{subfigure}
\usepackage{dsfont}
\usepackage[final]{graphicx}
\usepackage[T1]{fontenc}
\usepackage{amsmath}
\usepackage{mathtools, cuted}
\usepackage{amsthm}
\usepackage{amstext}
\usepackage{amssymb}
\usepackage{mathrsfs}
\usepackage{cite}
\usepackage{mathtools}
\usepackage{tikz}
\DeclareMathAlphabet{\mathpzc}{OT1}{pzc}{m}{it}
\usetikzlibrary{arrows,positioning, shapes}
\usepackage{float}

\usepackage{setspace}
\doublespace

\definecolor{Darkgreen}{rgb}{0,0.4,0}
\definecolor{wine-stain}{rgb}{0.5,0,0}
\usepackage{hyperref}
\hypersetup{colorlinks,
    colorlinks,
    citecolor=red,
    filecolor=green,
    linkcolor=blue,
    linktoc=page,
    urlcolor=blue
}

\def\R{\mathbb{R}}

\def\eps{\varepsilon}

\def\D{{\mathcal D}}
\def\E{{\mathbb E}}

\def\I{{\mathcal I}}

\def\C{{\mathcal C}}

\def\Q{{\mathcal Q}}

\def\X{{\mathcal X}}
\def\Y{{\mathcal Y}}
\def\M{{\mathcal M}}
\def\N{{\mathcal N}}

\def\L{{\mathcal L}}

\def\Z{{\mathcal Z}}
\def\U{{\mathcal U}}
\def\V{{\mathcal V}}
\def\indep{{\perp\!\!\!\perp}}

\def\sBer{{\mathsf{Bernoulli}}}

\def\sW{{\mathsf W}}

\def \var {{\mathsf {var}   }}
\def \mmse {{\mathsf {mmse}   }}
\usepackage[font=small,labelsep=space]{caption}
\DeclareCaptionLabelSeparator{dot}{.~}
\newcounter{example}
\newenvironment{example}[1][]{\refstepcounter{example}\par\medskip
   \noindent \textit{Example~\theexample. #1} \rmfamily}{\medskip}
\newtheorem{definition}{Definition}
\newtheorem{theorem}{Theorem}
\newtheorem{corollary}{Corollary}

\newtheorem{lemma}{Lemma}
\theoremstyle{remark}
\newtheorem{remark}{Remark}

\newcommand{\markov}{\mathrel\multimap\joinrel\mathrel-%
\mspace{-9mu}\joinrel\mathrel-}

\usepackage{times}
\usepackage{tikz}
\usepackage{amsmath}
\usepackage{verbatim}
\usetikzlibrary{arrows,shapes}
\tikzstyle{RectObject}=[rectangle,fill=white,draw,line width=0.2mm]
\tikzstyle{line}=[draw]
\tikzstyle{arrow}=[draw, -latex]
\usetikzlibrary{decorations.pathmorphing}
\usetikzlibrary{calc,shapes, positioning}
\usepackage{graphicx}
\usepackage{caption}
\usetikzlibrary{shapes.geometric}

\IEEEoverridecommandlockouts

\allowdisplaybreaks

\begin{document}

\title{\vspace{5.5mm}{Information Extraction Under Privacy Constraints\thanks{Parts of the results in this paper were presented at the 52nd Allerton Conference on Communications,
Control and Computing \cite{Asoodeh_Allerton} and the 14th Canadian Workshop on Information Theory \cite{Asoodeh_CWIT}.}}}
%\author{\IEEEauthorblockN{Shahab Asoodeh}
\author{\IEEEauthorblockN{Shahab Asoodeh, Mario Diaz, Fady Alajaji, and Tam\'{a}s Linder}\\
    \IEEEauthorblockA{\normalsize{Department of Mathematics and Statistics, Queen's University}
    %\\ asoodehshahab@mast.queensu.ca  }}
    \\\{asoodehshahab, 13madt, fady, linder\}@mast.queensu.ca}}
\restoregeometry
\maketitle

\begin{abstract}
A privacy-constrained information extraction problem is considered
where for a pair of correlated discrete random variables $(X,Y)$
governed by a given joint distribution, an agent observes $Y$ and
wants to convey to a potentially public user as much information
about $Y$ as possible without compromising the amount of information
revealed about $X$. To this end, the so-called {\em rate-privacy function}
is investigated to quantify the maximal amount of information (measured in
terms of mutual information) that can be extracted from $Y$ under a
privacy constraint between $X$ and the extracted information, where
privacy is measured using either mutual information or maximal
correlation. Properties of the rate-privacy function are analyzed
and information-theoretic and estimation-theoretic interpretations
of it are presented for both the mutual information and maximal
correlation privacy measures. It is also shown that the rate-privacy
function admits a closed-form expression for a large family of joint
distributions of $(X,Y)$. Finally, the rate-privacy function under the
mutual information privacy measure is considered for the case where
$(X,Y)$ has a joint probability density function by studying the problem
where the extracted information is a uniform quantization of $Y$ corrupted
by additive Gaussian noise. The asymptotic behavior of the
rate-privacy function is studied as the quantization resolution grows
without bound and it is observed that not all of the properties of the
rate-privacy function carry over from the discrete to the continuous
case.
\end{abstract}
\begin{IEEEkeywords}
Data privacy, equivocation, rate-privacy function, information theory, MMSE and additive channels, mutual information, maximal correlation.
\end{IEEEkeywords}
%\break
%{\hypersetup{linkbordercolor=black}
%\tableofcontents}

\section{Introduction}

With the emergence of user-customized services, there is an increasing desire to balance between the need to share data and the need to protect sensitive and private information. For example, individuals who join a social network are asked to provide information about themselves which might compromise their privacy. However, they agree to do so, to some extent, in order to benefit from the customized
services such as recommendations and personalized searches. As another example, a participatory technology for estimating road traffic requires each individual to provide her start and destination points as well as the travel time. However, most participating individuals prefer to provide somewhat distorted or false information to protect their privacy. Furthermore, suppose a software company wants to gather statistical information on how people use its software. Since many users might have used the software to handle some personal or sensitive information -for example, a browser for anonymous web
surfing or a financial management software- they may not want to share their data with the company. On the other hand, the company cannot legally collect the raw data either, so it needs to entice its users. In all these situations, a tradeoff in a conflict between utility advantage and privacy breach is required and the question is how to achieve this tradeoff. For example, how can a company collect high-quality aggregate information about users while strongly guaranteeing to its
users that it is not storing user-specific information?

To deal with such privacy considerations, Warner \cite{warner} proposed the \emph{randomized response model} in which each individual user randomizes her own data using a local randomizer (i.e., a noisy channel) before sharing the data to an untrusted data collector to be aggregated. As opposed to \emph{conditional security}, see e.g. \cite{computational_privacy, computational_privacy2, computational_privacy3}, the randomized response model assumes that the adversary can have unlimited computational power and thus it provides \emph{unconditional} privacy. This model, in which the control of private data remains in the users' hands, has been extensively studied since Warner. As a special case of the randomized response model, Duchi et al.\ \cite{privacyaware}, inspired by the well-known privacy guarantee called differential privacy introduced by Dwork et al.\ \cite{dwork1, dwork2, dwork3},  introduced locally differential privacy (LDP).  Given a random variable $X\in\X$, another random variable $Z\in\Z$ is said to be the $\eps$-LDP version of $X$ if there exists a channel $Q:X\to Z$ such that $\frac{Q(B|x)}{Q(B|x')}\leq \exp(\eps)$ for all measurable $B\subset \Z$ and all $x,x'\in\X$. The channel $Q$ is then called as the $\eps$-LDP mechanism. Using Jensen's inequality, it is straightforward to see that any $\eps$-LDP mechanism leaks at most $\eps$ bits of private information, i.e., the mutual information between $X$ and $Z$ satisfies $I(X,Z)\leq \eps$.

There have been numerous studies on the tradeoff between privacy and utility for different examples of randomized response models with different choices of utility and privacy measures. For instance, Duchi et al.\ \cite{privacyaware} studied the optimal  $\eps$-LDP mechanism $\M:X\to Z$ which minimizes the risk of estimation of a parameter $\theta$ related to $P_X$. Kairouz et al.\ \cite{Kairouz_PHD} studied an optimal $\eps$-LDP mechanism in the sense of mutual information, where an individual would like to release an $\eps$-LDP version $Z$ of $X$ that preserves as much information about $X$ as possible. Calmon et al.\ \cite{Calmon_bounds_Inference} proposed a novel privacy measure (which includes maximal correlation and chi-square correlation) between $X$ and $Z$ and studied the optimal privacy mechanism (according to their privacy measure) which minimizes the error probability $\Pr(\hat{X}(Z)\neq X)$ for any estimator $\hat{X}:Z\to X$.

In all above examples of randomized response models, given a private source, denoted by $X$, the mechanism generates $Z$ which can be publicly displayed without breaching the desired privacy level. However, in a more realistic model of privacy, we can assume that for any given private data $X$, nature generates $Y$, via a fixed channel $P_{Y|X}$. Now we aim to release a public display $Z$ of $Y$ such that the amount of information in $Y$ is preserved as much as possible while $Z$ satisfies a privacy constraint with respect to $X$. Consider two communicating agents Alice and Bob. Alice collects all her measurements from an observation into a random variable $Y$ and ultimately wants to reveal this information to Bob in order to receive a payoff. However, she is worried about her private data, represented by $X$, which is correlated with $Y$. For instance, $X$ might represent her precise location and $Y$ represents  measurement of traffic load of a route she has taken. She wants to reveal these measurements to an online road monitoring system to received some utility. However, she does not want to reveal too much information about her exact location. In such situations, the utility is measured  with respect to $Y$ and privacy is measured with respect to $X$. The question raised in this situation then concerns the maximum payoff Alice can get from Bob (by revealing $Z$ to him) without compromising her privacy. Hence, it is of interest to characterize such competing objectives in the form of a quantitative tradeoff. Such a characterization provides a controllable balance between utility and privacy.

This model of privacy first appears in Yamamoto's work \cite{yamamotoequivocationdistortion} in which the rate-distortion-equivocation function is defined as the tradeoff between a distortion-based utility and privacy. Recently, Sankar et al.\ \cite{Lalitha_Forensics}, using the quantize-and-bin scheme \cite{Tandon_Quantize_bin}, generalized Yamamoto's model to study privacy in databases from an information-theoretic point of view. Calmon and Fawaz \cite{Calmon_privacy_Aganist} and Monedero et al.\ \cite{t_closeness} also independently used distortion and mutual information for utility and privacy, respectively, to define a privacy-distortion function which resembles the classical rate-distortion function. More recently, Makhdoumi et al.\ \cite{Funnel} proposed to use mutual information for both utility and privacy measures and defined the \emph{privacy funnel} as the corresponding privacy-utility tradeoff, given by
\begin{equation}\label{Dual_gEpsilon}
  t_{R}(X;Y):=\min_{\substack{P_{Z|Y}:X\markov Y\markov Z\\I(Y; Z)\geq R}} I(X;Z),
\end{equation}
where $X\markov Y\markov Z$ denotes that $X, Y$ and $Z$ form a Markov chain in this order. Leveraging well-known algorithms for the information bottleneck problem \cite{information_bottleneck}, they provided a locally optimal greedy algorithm to evaluate $t_R(X;Y)$. Asoodeh et al.\ \cite{Asoodeh_Allerton}, independently, defined the \emph{rate-privacy function}, $g_{\eps}(X;Y)$, as the maximum achievable $I(Y; Z)$ such that $Z$ satisfies $I(X;Z)\leq\eps$, which is a dual representation of the privacy funnel \eqref{Dual_gEpsilon}, and showed that for discrete $X$ and $Y$, $g_0(X;Y)>0$ if and only if $X$ is \emph{weakly independent} of $Y$ (cf, Definition~\ref{definition_weakly_ind}). Recently, Calmon et al.\ \cite{Calmon_fundamental-Limit} proved an equivalent result for $t_R(X;Y)$ using a different approach. They also obtained lower and upper bounds for $t_{R}(X;Y)$ which can be easily translated to bounds for $g_{\eps}(X;Y)$ (cf. Lemma\ref{non-increasing-Lemma}). In this paper, we develop further properties of $g_{\eps}(X;Y)$ and also determine necessary and sufficient conditions on $P_{XY}$, satisfying some symmetry conditions, for $g_{\eps}(X;Y)$ to achieve its upper and lower bounds.

The problem treated in this paper can also be contrasted with the better-studied concept of \emph{secrecy} following the pioneering work of Wyner \cite{Wyner_wiretap}. While in secrecy problems the aim is  to keep information secret only from wiretappers, in privacy problems the aim is to keep the private information (not necessarily all the information) secret from everyone including the intended receiver.

\subsection{Our Model and Main Contributions}

 Using mutual information as measure of both utility and privacy, we formulate the corresponding privacy-utility tradeoff for discrete random variables $X$ and $Y$ via the rate-privacy function, $g_{\eps}(X;Y)$, in which the mutual information between $Y$ and displayed data (i.e., the mechanism's output), $Z$, is maximized over all channels $P_{Z|Y}$ such that the mutual information between $Z$ and $X$ is no larger than a given $\eps$. We also formulate a similar rate-privacy function $\hat{g}_{\eps}(X;Y)$ where the privacy is measured in terms of the squared maximal correlation, $\rho_m^2$, between,  $X$ and $Z$. In studying $g_{\eps}(X;Y)$ and $\hat{g}_{\eps}(X;Y)$, any channel $Q:Y\to Z$ that satisfies $I(X;Z)\leq \eps$ and $\rho_m^2(X;Z)\leq \eps$, preserves the desired level of privacy and is hence called a \emph{privacy filter}.
  Interpreting $I(Y;Z)$ as the number of bits that a privacy filter can reveal about $Y$ without compromising privacy, we present the rate-privacy function as a formulation of the problem of maximal \emph{privacy-constrained information extraction} from $Y$.

  We remark that using maximal correlation as a privacy measure is by no means new as it appears in other works, see e.g., \cite{Fawaz_Makhdoumi}, \cite{MaximalCorr_Secrecy} and \cite{Calmon_bounds_Inference} for different utility functions.  We do not put any likelihood constraints on the privacy filters as opposed to the definition of LDP. In fact, the optimal privacy filters that we obtain in this work induce channels $P_{Z|X}$ that do not satisfy the LDP property.

 The quantity $g_{\eps}(X;Y)$ is related  to a notion of the \emph{reverse} strong data processing inequality as follows. Given a joint distribution $P_{XY}$, the strong data processing coefficient was introduced in \cite{Ahlswede_Gacs} and \cite{Anantharam}, as the smallest $s(X;Y)\leq 1$ such that $I(X;Z)\leq s(X;Y) I(Y;Z)$ for all $P_{Z|Y}$ satisfying the Markov condition $X\markov Y\markov Z$. In the rate-privacy function, we instead seek an upper bound on the maximum achievable rate at which $Y$ can display information, $I(Y;Z)$, while meeting the privacy constraint $I(X;Z)\leq \eps$.  The connection between the rate-privacy function and the strong data processing inequality is further studied in \cite{Calmon_fundamental-Limit} to mirror all the results of \cite{Anantharam} in the context of privacy.

 The contributions of this work are as follows:
 \begin{itemize}
  \item We study lower and upper bounds of $g_{\eps}(X;Y)$.  The lower bound, in particular, establishes a multiplicative bound on $I(Y;Z)$ for any optimal privacy filter. Specifically,  we show that for a given $(X,Y)$ and $\eps>0$ there exists a channel $Q:Y\to Z$ such that $I(X;Z)\leq \eps$ and
         \begin{equation}\label{intuitive-Definition_Gepsilon}
       I(Y;Z)\geq \lambda(X;Y) \eps,
     \end{equation}
 where $\lambda(X;Y)\geq 1$ is a constant depending on the joint distribution $P_{XY}$.
 We then give conditions on $P_{XY}$ such that the upper and lower bounds are tight. For example, we show that the lower bound is achieved when $Y$ is binary and the channel from $Y$ to $X$ is symmetric. We show that this corresponds to the fact that both $Y=0$ and $Y=1$ induce distributions $P_{X|Y}(\cdot|0)$ and $P_{X|Y}(\cdot|1)$ which are equidistant from $P_X$ in the sense of Kullback-Leibler divergence. We then show that the upper bound is achieved when $Y$ is an erased version of $X$, or equivalently, $P_{Y|X}$ is an erasure channel.

 \item We propose an information-theoretic setting in which $g_{\eps}(X;Y)$ appears as a natural upper-bound for the achievable rate in the so-called "dependence dilution" coding problem.
     Specifically, we examine the joint-encoder version of an \emph{amplification-masking tradeoff}, a setting recently introduced by Courtade \cite{Courtade_Aplification} and we show that the dual of $g_{\eps}(X;Y)$ upper bounds the masking rate. We also present an estimation-theoretic motivation for the privacy measure $\rho_m^2(X;Z)\leq \eps$. In fact, by imposing $\rho_m^2(X;Y)\leq \eps$, we require that an adversary who observes $Z$ cannot efficiently estimate $f(X)$, for any function $f$. This is reminiscent of \emph{semantic security} \cite{Goldwasser1984270} in the cryptography community. An encryption mechanism is said to be semantically secure if the adversary's advantage for correctly guessing \emph{any function} of the privata data given an observation of the mechanism's output (i.e., the ciphertext) is required to be negligible. This, in fact, justifies the use of maximal correlation as a measure of privacy.  The use of mutual information as privacy measure can also be justified using Fano's inequality. Note that $I(X;Z)\leq \eps$ can be shown to imply that $\Pr(\hat{X}(Z)\neq X)\geq \frac{H(X)-1-\eps}{\log(|\X|)}$ and hence the probability of adversary correctly guessing $X$ is lower-bounded.

\item We also study the rate of increase $g'_0(X;Y)$ of $g_{\eps}(X;Y)$ at $\eps=0$ and show that this rate can characterize the behavior of $g_{\eps}(X;Y)$ for any $\eps\geq 0$ provided that $g_0(X;Y)=0$. This again has connections with the results  of \cite{Anantharam}. Letting
    $$\Gamma(R):=\max_{P_{Z|Y}:X\markov Y \markov Z\atop I(Y;Z)\leq R}I(X;Z),$$ one can easily show that $\Gamma'(0)=\lim_{R\to 0}\frac{\Gamma(R)}{R}=s(X;Y),$ and hence the rate of increase of $\Gamma(R)$ at $R=0$ characterizes the strong data processing coefficient. Note that here we have $\Gamma(0)=0$.

\item Finally, we generalize the rate-privacy function to the continuous case where $X$ and $Y$ are both continuous and show that some of the properties of $g_{\eps}(X;Y)$ in the discrete case do not carry over to the continuous case. In particular, we assume that the privacy filter belongs to a family of additive noise channels followed by an $M$-level uniform scalar quantizer and give asymptotic bounds as $M\to \infty$ for the rate-privacy function.

\end{itemize}
\subsection{Organization}

 The rest of the paper is organized as follows. In Section 2, we define and study the rate-privacy function for discrete random variables for two different privacy measures, which, respectively, lead to the information-theoretic and estimation-theoretic interpretations of the rate-privacy function. In Section~3, we provide such interpretations for the rate-privacy function in terms of quantities from information and estimation theory.  Having obtained lower and upper bounds of the rate-privacy function, in Section~4 we determine the conditions on $P_{XY}$ such that these bounds are tight. The rate-privacy function is then generalized and studied  in Section~5 for continuous random variables.

\section{Utility-Privacy Measures: Definitions and Properties}

Consider two random variables $X$ and $Y$, defined over \emph{finite} alphabets $\X$ and $\Y$, respectively, with a fixed joint distribution $P_{XY}$. Let $X$ represent the \emph{private data} and let $Y$ be the \emph{observable data}, correlated with $X$ and generated by the channel $P_{Y|X}$ predefined by nature, which we call the \emph{observation channel}. Suppose there exists a channel $P_{Z|Y}$ such that $Z$, the \emph{displayed data} made available to public users, has limited dependence with $X$. Such a channel is called the \emph{privacy filter}. This setup is shown in Fig.~\ref{fig:privacy}. The objective is then to find a privacy filter which gives rise to the highest dependence between $Y$ and $Z$. To make this goal precise, one needs to specify a measure for both utility (dependence between $Y$ and $Z$) and also privacy (dependence between $X$ and $Z$).

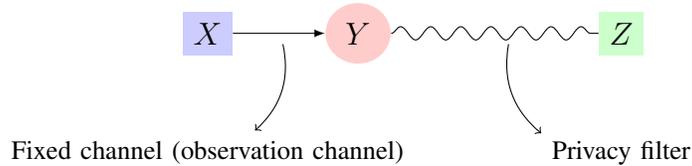
\begin{figure}[!h]
\centering
\begin{tikzpicture}
        \draw (-1,-1) node[fill=blue!20, anchor=base] (private) {$X$};
        \draw (1,-1) node[fill=red!20, ellipse,anchor=base] (public) {$Y$};
        \draw (4.5,-1) node[fill=green!20, anchor=base] (generate) {$Z$};
        \path [arrow] (private) -- (public);
        \draw[decorate, decoration = {snake, segment length = .4cm }] (public) --  (generate);
         \node[below = of private] (note1) {\small{Fixed channel (observation channel)}};
          \coordinate (channel1) at (0,-1);
         \draw [->] (channel1)        to [bend left] (note1);
          \node[below  = of generate ] (note2) {\small{Privacy filter}};
          \coordinate (channel2) at (3,-1);
         \draw [->] (channel2)        to [bend right] (note2);
\end{tikzpicture}
\caption{\small{Information-theoretic privacy.}} \label{fig:privacy}
\end{figure}

 \subsection{Mutual Information as Privacy Measure}

 Adopting mutual information as a measure of  both privacy and utility, we are interested in characterizing the following quantity, which we call the \emph{rate-privacy function}\footnote{Since mutual information is adopted for utility, the privacy-utility tradeoff characterizes the optimal \emph{rate} for a given privacy level, where rate indicates the precision of the displayed data $Z$ with respect to the observable data $Y$ for a privacy filter, which suggests the name.},
\begin{equation}\label{gepsilon}
  g_{\eps}(X;Y):=\sup_{P_{Z|Y}\in \D_{\eps}(P)}I(Y;Z),
\end{equation}
where $(X, Y)$ has fixed distribution $P_{XY}=P$ and $$\D_{\eps}(P):=\{P_{Z|Y}:X\markov Y\markov Z, I(X;Z)\leq \eps\},$$ (here $X\markov Y\markov Z$ means that $X,Y,$ and $Z$ form a Markov chain in this order). Equivalently, we call $g_{\eps}(X;Y)$ the \emph{privacy-constrained information extraction function}, as $Z$ can be thought of as the extracted information from $Y$ under privacy constraint $I(X;Z)\leq \eps$.

Note that since $I(Y;Z)$ is a convex function of $P_{Z|Y}$ and furthermore the constraint set $\D_{\eps}(P)$ is convex, \cite[Theorem 32.2]{rock} implies that we can restrict $\D_{\eps}(P)$ in \eqref{gepsilon} to $\{P_{Z|Y}:X\markov Y\markov Z, I(X;Z)= \eps\}$ whenever $\eps\leq I(X;Y)$ . Note also that since for finite $\X$ and $\Y$, $P_{Z|Y}\to I(Y; Z)$ is a continuous map, therefore $\D_{\eps}(P)$ is  compact and the supremum in \eqref{gepsilon} is indeed a maximum. In this case, using the Support Lemma \cite{csiszarbook}, one can readily show that it suffices that the random variable $Z$ is supported on an alphabet $\Z$ with cardinality $|\Z|\leq |\Y|+1$. Note further that by the Markov condition $X\markov Y\markov Z$, we can always restrict $\eps\geq 0$ to only $0\leq\eps<I(X;Y)$, because $I(X;Z)\leq I(X;Y)$ and hence for $\eps\geq I(X;Y)$ the privacy constraint is removed and thus by setting $Z=Y$, we obtain $g_{\eps}(X;Y)=H(Y)$.

As mentioned earlier, a dual representation of $g_{\eps}(X;Y)$, the so called \emph{privacy funnel}, is introduced in \cite{Funnel} and \cite{Calmon_fundamental-Limit}, defined in \eqref{Dual_gEpsilon}, as the least information leakage about $X$ such that the communication rate is greater than a positive constant; $I(Y; Z)\geq R$ for some $R>0$. Note that if $t_{R}(X;Y)=\eps$ then $g_{\eps}(X;Y)=R$.
%\begin{equation}\label{Dual_gEpsilon}
%  t_{R}(X;Y):=\min_{\substack{P_{Z|Y}:X~\markov ~~Y~\markov ~~Z\\I(Y; Z)\geq R}} I(X;Z).
%\end{equation}

%The rate-privacy function can equivalently be characterized using a \emph{more capable} \cite{comparisonoftwochannels} ordering on channels in $\D_{\eps}(P)$.  Given two channels $\sW_1: Y\to Z_1$ and $\sW_2:Y\to Z_2$, $\sW_1$ is said to be more capable than $\sW_2$, denoted by $\sW_2\prec\sW_1$, if $I(Y;Z_2)\leq I(Y;Z_1)$ for any input distribution $P_Y$. Restricting this definition to channels in $\D_{\eps}(P)$, we define \emph{more privately capable} ordering, denoted by $\sW_1 \underset{P}{\overset{\eps}{\prec}} \sW_2$, when  $\sW_2\prec\sW_1$ and both $\sW_1$ and $\sW_2$ are privacy filter, that is, $\sW_1\in \D_{\eps}(P)$ and $\sW_2\in \D_{\eps}(P)$. We can equivalently characterize $g_{\eps}(X;Y)$ as $g_{\eps}(P)=I(P_Y, \sW^*)$ where $P=P_{XY}$ whose $Y$-marginal is $P_Y$ and $\sW^*$ is the \emph{most privately capable} channel, i.e., $\sW \underset{P}{\overset{\eps}{\prec}} \sW^*$ for any privacy filter $\sW$.

%In fact, one can define another channel ordering as follows. Given a joint distribution $P=P_{XY}$ and $\eps\geq 0$, $\sW_1:Y\to Z_1$ is said to be \emph{more privately capable} than $\sW_2:Y\to Z_2$, denoted by $\sW_2\stackrel{\eps}{\preceq} \sW_1$, if  for $X\markov Y\markov (Z_1, Z_2)$, we have $I(X;Z_1)\leq \eps$, $I(X;Z_2)\leq \eps$, and $I(Y; Z_1)\leq I(P; Z_2)$.
Given $\eps_1< \eps_2$ and a joint distribution $P=P_{X}\times P_{Y|X}$, we have $\D_{\eps_1}(P)\subset \D_{\eps_2}(P)$ and hence $\eps\to g_{\eps}(X;Y)$ is non-decreasing, i.e., $g_{\eps_1}(X;Y)\leq g_{\eps_2}(X;Y)$. Using a similar technique as in \cite[Lemma 1]{schulman}, Calmon et al.\ \cite{Calmon_fundamental-Limit} showed that the mapping $R\mapsto \frac{t_{R}(X; Y)}{R}$ is non-decreasing for $R>0$. This, in fact, implies that $\eps\mapsto \frac{g_{\eps}(X; Y)}{\eps}$ is non-increasing for $\eps>0$. This observation leads to a lower bound for the rate privacy function $g_{\eps}(X;Y)$ as described in the following lemma.
\begin{lemma}[\cite{Calmon_fundamental-Limit}]\label{non-increasing-Lemma}
For a given joint distribution $P$ defined over $\X\times \Y$, the mapping $\eps\mapsto \frac{g_{\eps}(X; Y)}{\eps}$ is non-increasing on $\eps\in (0, \infty)$ and $g_{\eps}(X;Y)$ lies between two straight lines as follows:
\begin{equation}\label{UB&LB}
    \eps\frac{H(Y)}{I(X;Y)}\leq g_{\eps}(X;Y)\leq H(Y|X)+\eps,
\end{equation}
for $\eps\in(0, I(X;Y))$.
\end{lemma}
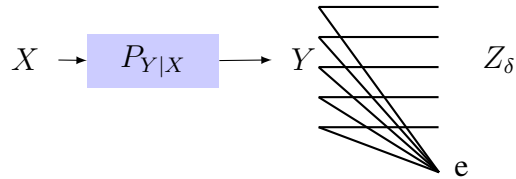
\begin{figure}[H]
\centering
\begin{tikzpicture}
        \node (x) [circle] at (-3.7,-0.9) {$X$};
        \draw (-2,-1) node[fill=blue!20, anchor=base] (channel) {$~~P_{Y|X}~~$};
        \path [arrow] (x) -- (channel);
        \node (y) [circle] at (0,-0.9) {$Y$};
        \path [arrow] (channel) -- (y);
        \node (y1) [circle] at (0.15,-0.2) {};
        \node (y2) [circle] at (0.15,-0.6) {};
        \node (y3) [circle] at (0.15,-1) {};
        \node (y4) [circle] at (0.15,-1.4) {};
        \node (y5) [circle] at (0.15,-1.8) {};

        \node (z1) [circle] at (2,-0.2) {};
        \node (z2) [circle] at (2,-0.6) {};
        \node (z3) [circle] at (2,-1) {};
        \node (z4) [circle] at (2,-1.4) {};
        \node (z5) [circle] at (2,-1.8) {};
        \node (z6) [circle] at (2,-2.4) {};
        \node (z8) [circle] at (2.1,-2.35) {e};
        \node (z7) [circle] at (2.6,-0.9) {$Z_{\delta}$};
        \draw[thick]  (0.2,-0.2) -- (z1);
        \draw[thick]  (0.2,-0.2) -- (1.8,-2.4);
        \draw[thick]  (0.2,-0.6) -- (z2);
        \draw[thick]  (0.2,-0.6) -- (1.8,-2.4);
        \draw[thick]  (0.2,-1) -- (z3);
        \draw[thick]  (0.2,-1) -- (1.8,-2.4);
        \draw[thick]  (0.2,-1.4) -- (z4);
        \draw[thick]  (0.2,-1.4) -- (1.8,-2.4);
        \draw[thick]  (0.2,-1.8) -- (z5);
        \draw[thick]  (0.2,-1.8) -- (1.8,-2.4);
\end{tikzpicture}
\caption{\small{Privacy filter that achieves the lower bound in \eqref{UB&LB} where $Z_{\delta}$ is the output of an erasure privacy filter with erasure probability specified in \eqref{Delta}.}} \label{fig:LowerBound_Filter}
\end{figure}
Using a simple calculation, the lower bound in \eqref{UB&LB} can be shown to be achieved by the privacy filter depicted in Fig.~\ref{fig:LowerBound_Filter} with the erasure probability
\begin{equation}\label{Delta}
  \delta=1-\frac{\eps}{I(X;Y)}.
\end{equation}
In light of Lemma~\ref{non-increasing-Lemma}, the possible range of the map $\eps\mapsto g_{\eps}(X;Y)$ is as depicted in Fig.~\ref{Fig: UB&LB}.
\begin{figure}
  \centering
  \includegraphics[width=12cm, height=8cm]{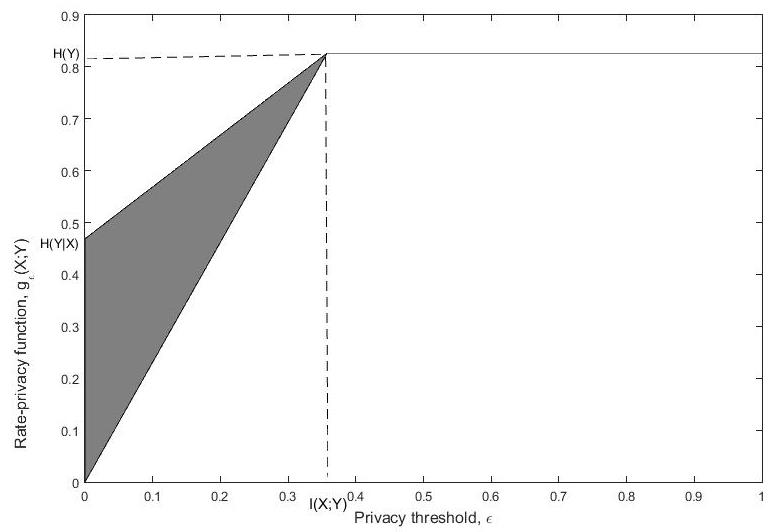}\\
  \caption{\small{The region of $g_{\eps}(X;Y)$ in terms of $\eps$ specified by \eqref{UB&LB}.}}\label{Fig: UB&LB}
\end{figure}
We next show that $\eps\mapsto g_{\eps}(X;Y)$ is concave and continuous.
\begin{lemma}\label{Lemma_Concavity}
For any given pair of random variables $(X,Y)$ over $\X\times \Y$, the mapping $\eps\mapsto g_{\eps}(X;Y)$ is concave for $\eps\geq 0$.
\end{lemma}
\begin{proof}
It suffices to show that for any $0\leq \eps_1<\eps_2<\eps_3\leq I(X;Y)$, we have
\begin{equation}\label{concavity_requirement}
    \frac{g_{\eps_3}(X;Y)-g_{\eps_1}(X;Y)}{\eps_3-\eps_1}\leq \frac{g_{\eps_2}(X;Y)-g_{\eps_1}(X;Y)}{\eps_2-\eps_1},
\end{equation}
which, in turn, is equivalent to
\begin{equation}\label{concavity_requirement_2}
\left(\frac{\eps_2-\eps_1}{\eps_3-\eps_1}\right)g_{\eps_3}(X;Y)+\left(\frac{\eps_3-\eps_2}{\eps_3-\eps_1}\right)g_{\eps_1}(X;Y)\leq g_{\eps_2}(X;Y).\end{equation}
Let $P_{Z_1|Y}: Y\to Z_1$   and $P_{Z_3|Y}:Y\to Z_3$ be two optimal privacy filters in $\D_{\eps_1}(P)$ and $\D_{\eps_3}(P)$ with disjoint output alphabets $\Z_1$ and $\Z_3$, respectively.

We introduce an auxiliary binary random variable $U\sim\sBer(\lambda)$, independent of $(X,Y)$, where $\lambda:=\frac{\eps_2-\eps_1}{\eps_3-\eps_1}$ and define the following random privacy filter $P_{Z_{\lambda}|Y}$: We pick $P_{Z_3|Y}$ if $U=1$ and $P_{Z_1|Y}$ if $U=0$, and let $Z_{\lambda}$ be the output of this random channel which takes values in $\Z_1\cup \Z_3$. Note that $(X,Y)~\markov Z\markov~ U$. Then we have
\begin{eqnarray*}
  I(X;Z_{\lambda}) &=& I(X;Z_{\lambda},U)= I(X;Z_{\lambda}|U)=\lambda I(X;Z_3)+(1-\lambda)I(X;Z_1), \\
   &\leq& \eps_2,
\end{eqnarray*}
which implies that $P_{Z_{\lambda}|Y}\in\D_{\eps_2}(P)$. On the other hand,  we have
\begin{eqnarray*}
  g_{\eps_2}(X;Y)\geq I(Y;Z_{\lambda}) &=& I(Y;Z_{\lambda},U)= I(Y;Z_{\lambda}|U)=\lambda I(Y;Z_3)+(1-\lambda)I(Y;Z_1), \\
   &=&\left(\frac{\eps_2-\eps_1}{\eps_3-\eps_1}\right)g_{\eps_3}(X;Y)+\left(\frac{\eps_3-\eps_2}{\eps_3-\eps_1}\right)g_{\eps_1}(X;Y)
\end{eqnarray*}
which, according to \eqref{concavity_requirement_2}, completes the proof.
\end{proof}
\begin{remark}
By the concavity of $\eps\mapsto g_{\eps}(X;Y)$, we can show that $g_{\eps}(X;Y)$ is a \emph{strictly} increasing function of $\eps\leq I(X;Y)$. To see this, assume there exists $\eps_1< \eps_2\leq I(X;Y)$ such that $g_{\eps_1}(X;Y)=g_{\eps_2}(X;Y)$. Since $\eps\mapsto g_{\eps}(X;Y)$ is concave, then it follows that for all $\eps\geq \eps_2$, $g_{\eps}(X;Y)=g_{\eps_2}(X;Y)$ and since for $\eps=I(X;Y)$, $g_{I(X;Y)}(X;Y)=H(Y)$, implying that for any $\eps\geq \eps_2$, we must have $g_{\eps}(X;Y)=H(Y)$ which contradicts the upper bound shown in \eqref{UB&LB}.
\end{remark}

\begin{corollary}\label{Corollary_Continuity}
For any given pair of random variables $(X,Y)$ over $\X\times \Y$, the mapping $\eps\mapsto g_{\eps}(X;Y)$ is continuous for $\eps\geq 0$.
\end{corollary}
\begin{proof}
Concavity directly implies that the mapping $\eps\mapsto g_{\eps}(X;Y)$ is continuous on $(0, \infty)$ (see for example \cite[Theorem 3.2]{Rudin}). Continuity at zero follows from the continuity of mutual information. %In fact, we know that $\D_{\eps}(P_{XY})$ shrinks down to $\D_0(P_{XY})$ when $\eps\to 0$ and hence $g_{\eps}(X;Y)\to g_0(X;Y)$.
\end{proof}
\begin{remark}\label{Remark_General_LoerBound}
Using the concavity of the map $\eps\mapsto g_{\eps}(X;Y)$, we can provide an alternative proof for the lower bound in \eqref{UB&LB}. Note that point $(I(X;Y), H(Y))$ is always on the curve $g_{\eps}(X;Y)$, and hence by concavity, the straight line $\eps\mapsto \eps\frac{H(Y)}{I(X;Y)}$ is always below the lower convex envelop of $g_{\eps}(X;Y)$, i.e., the chord connecting $(0, g_0(X;Y))$ to $(I(X;Y), H(Y))$,  and hence $g_{\eps}(X;Y)\geq\eps\frac{H(Y)}{I(X;Y)}$. In fact, this chord yields a better lower bound for $g_{\eps}(X;Y)$ on $\eps\in [0, I(X;Y]$ as
\begin{equation}\label{General_Lower_Bound}
  g_{\eps}(X;Y)\geq \eps\frac{H(Y)}{I(X;Y)}+g_0(X;Y)\left[1-\frac{\eps}{I(X;Y)}\right],
\end{equation}
which reduces to the lower bound in \eqref{UB&LB} only if $g_0(X;Y)=0$.
\end{remark}
\subsection{Maximal Correlation as Privacy Measure}

By adopting the mutual information as the privacy measure between the private and the displayed data, we make sure that only limited bits of private information is revealed during the process of transferring $Y$. In order to have an estimation theoretic guarantee of privacy, we propose alternatively to measure  privacy using a \emph{measure of correlation}, the so-called maximal correlation.

Given the collection $\C$of all pairs of random variables $(U,V)\in\U\times \V$ where $\U$ and $\V$ are general alphabets, a mapping $T:\C\to [0,1]$ defines a measure of correlation \cite{gebelien} if $T(U,V)=0$ if and only if $U$ and $V$ are independent (in short, $U\indep V$) and $T(U,V)$ attains its maximum value if $X=f(Y)$ or $Y=g(X)$ almost surely for some measurable real-valued functions $f$ and $g$. There are many different examples of measures of correlation including the Hirschfeld-Gebelein-R\'{e}nyi  maximal correlation \cite{hirschfild, gebelien, Renyi-dependence-measure}, the information measure \cite{Linfoot_information}, mutual information and $f$-divergence \cite{Csiszar_f_divergence}.
\begin{definition}[\cite{Renyi-dependence-measure}]
Given  random  variables $X$ and $Y$,  the  maximal  correlation\footnote{Recall that the correlation coefficient between $U$ and $V$, is defined as $\rho(U;V):=\frac{\text{cov}(U;V)}{\sigma_U\sigma_V}$, where $\text{cov}(U;V), \sigma_U$ and $\sigma_V$ are the covariance between $U$ and $V$, the standard deviations of $U$ and $V$, respectively.}  $\rho_m(X;Y)$ is defined as follows:
$$\rho_m(X;Y):= \sup_{f, g} \rho(f(X), g(Y))=\sup_{(f(X),g(Y))\in \mathcal{S}}\E[f(X)g(Y)],$$ where $\mathcal{S}$ is the collection of pairs of real-valued random variables $f(X)$ and $g(Y)$ such that $\E f(X)=\E g(Y)=0$ and $\E f^2(X)=\E g^2(Y)=1$. If $\mathcal{S}$ is empty (which happens precisely when at least one of $X$ and $Y$ is constant almost surely) then one defines $\rho_m(X;Y)$ to be 0. R\'{e}nyi \cite{Renyi-dependence-measure} derived an equivalent characterization of maximal correlation as follows:
\begin{equation}\label{Maximal_correlation_Equivalent}
  \rho^2_m(X;Y)=\sup_{f:\E f(X)=0, \E f^2(X)=1}\E\left[\E^2[f(X)|Y]\right].
\end{equation}
\end{definition}

Measuring privacy in terms of maximal correlation, we propose
$$\hat{g}_{\eps}(X;Y):=\sup_{P_{Z|Y}\in \hat{\D}_{\eps}(P)} I(Y; Z),$$ as the corresponding rate-privacy tradeoff, where $$\hat{\D}_{\eps}(P):=\{P_{Z|Y}:~X\markov Y\markov Z, ~\rho_m^2(X;Z)\leq \eps, P_{XY}=P\}.$$ Again, we equivalently call $\hat{g}_{\eps}(X;Y)$ as the privacy-constrained information extraction function, where here the privacy is guaranteed by $\rho_m^2(X;Z)\leq \eps$.

Setting $\eps=0$ corresponds to the case where $X$ and $Z$ are required to be statistically independent, i.e., absolutely no information leakage about the private source $X$ is allowed. This is called \emph{perfect privacy.}
Since the independence of $X$ and $Z$ is equivalent to $I(X;Z)=\rho_m(X;Z)=0$, we  have $\hat{g}_{0}(X;Y)=g_{0}(X;Y)$. However, for $\eps>0$, both $g_{\eps}(X;Y)\leq \hat{g}_{\eps}(X;Y)$ and $g_{\eps}(X;Y)\geq \hat{g}_{\eps}(X;Y)$ might happen in general. For general $\eps\geq 0$, it directly follows using \cite[Proposition 1]{MaximalCorr_Secrecy} that
$$\hat{g}_{\eps}(X;Y)\leq g_{\eps'}(X;Y),$$ where $\eps':=\log(k\eps+1)$ and $k:=|\X|-1$.

Similar to $g_{\eps}(X;Y)$, we see that for $\eps_1\leq \eps_2$,  $\hat{\D}_{\eps_1}(P)\subset \hat{\D}_{\eps_2}(P)$ and hence $\eps\to \hat{g}_{\eps}(X;Y)$ is non-decreasing. The following lemma is a counterpart of Lemma~\ref{non-increasing-Lemma} for $\hat{g}_{\eps}(X;Y)$.
\begin{lemma}\label{non-increasing-Lemma_g_hat}
For a given joint distribution $P_{XY}$ defined over $\X\times \Y$, $\eps\mapsto \frac{\hat{g}_{\eps}(X;Y)}{\eps}$ is non-increasing on $(0, \infty)$.
\end{lemma}
\begin{proof}
Like Lemma~\ref{non-increasing-Lemma}, the proof is similar to the proof of \cite[Lemma 1]{schulman}. We, however, give a brief proof for the sake of completeness.

For a given channel $P_{Z|Y}\in\hat{\D}_{\eps}(P)$ and $\delta\geq 0$, we can define a new channel with an additional symbol $e$ as follows
\begin{equation} \label{non-increasing-Lemma_Cases}
P_{Z'|Y}(z'|y) = \begin{cases}
(1-\delta)P_{Z|Y}(z'|y) &\text{if  $z'\neq e$}\\
\delta &\text{if $z'=e$ }
\end{cases}
\end{equation}
It is easy to check that $I(Y; Z')=(1-\delta)I(Y; Z)$ and also $\rho_m^2(X; Z')=(1-\delta)\rho_m^2(X; Z)$; see \cite[Page 8]{Lei_Zhai_PhD_Thesis}, which implies that $P_{Z'|Y}\in\hat{\D}_{\eps'}(P)$
where $\eps'=(1-\delta)\eps$. Now suppose that $P_{Z|Y}$ achieves $\hat{g}_{\eps}(X;Y)$, that is, $\hat{g}_{\eps}(X;Y)=I(Y; Z)$ and $\rho_m^2(X;Z)=\eps$. We can then write
$$\frac{\hat{g}_{\eps}(X;Y)}{\eps}=\frac{I(Y;Z)}{\eps}=\frac{I(Y; Z')}{\eps'}\leq \frac{g_{\eps'}(X;Y)}{\eps'}.$$ Therefore, for $\eps'\leq  \eps$ we have $\frac{g_{\eps'}(X;Y)}{\eps'}\geq \frac{g_{\eps}(X;Y)}{\eps}$.
\end{proof}

Similar to the lower bound for $g_{\eps}(X;Y)$ obtained from Lemma~\ref{non-increasing-Lemma}, we can obtain a lower bound for $\hat{g}_{\eps}(X;Y)$ using Lemma~\ref{non-increasing-Lemma_g_hat}. Before we get to the lower bound, we need a data processing lemma for maximal correlation. The following lemma proves a version of \emph{strong} data processing inequality for maximal correlation from which the typical data processing inequality follows, namely, $\rho_m(X;Z)\leq \min \{\rho_m(Y;Z),\rho_m(X;Y)\}$ for $X,Y$ and $Z$ satisfying $X\markov Y\markov Z$.
\begin{lemma}\label{Lemma_Data_Processing_Maximal_corr}
For random variables $X$ and $Y$ with a joint distribution $P_{XY}$, we have
$$\sup_{\substack{X\markov Y\markov Z\\ \rho_m(Y;Z)\neq 0}}\frac{\rho_m(X;Z)}{\rho_m(Y;Z)}=\rho_m(X;Y).$$
\end{lemma}
\begin{proof}
For arbitrary zero-mean and unit variance  measurable functions $f\in \L^2(\X)$ and $g\in\L^2(\Z)$ and $X\markov Y\markov Z$, we have
$$\E[f(X)g(Z)]=\E\left[\E[f(X)|Y]\E[g(Z)|Y]\right]\leq \rho_m(X;Y)\rho_m(Y;Z),$$
where the inequality follows from the Cauchy-Schwartz inequality and \eqref{Maximal_correlation_Equivalent}. Thus we obtain $\rho_m(X;Z)\leq\rho_m(X;Y)\rho_m(Y;Z)$.

This bound is tight for the special case of $X\to Y\to X'$, where $P_{X'|Y}$ is the backward channel associated with $P_{Y|X}$. In the following, we shall show that $\rho_m(X;Y)\rho_m(Y; X')=\rho_m(X;X')$.

To this end, first note that the above implies that $\rho_m(X;Y)\rho_m(Y;X')\geq \rho_m(X;X')$. Since $P_{XY}=P_{X'Y}$, it follows that $\rho_m(X;Y)=\rho_m(Y;X')$ and hence the above implies that $\rho_m^2(X;Y)\geq \rho_m(X;X')$. One the other hand, we have
$$\E[[\E[f(X)|Y]]^2]=\E[\E[f(X)|Y]\E[f(X')|Y]]=\E[\E[f(X)f(X')|Y]]=\E[f(X)f(X')],$$ which together with \eqref{Maximal_correlation_Equivalent} implies that
$$\rho^2_m(X;Y)\leq \sup_{f:\E f(X)=0, \E f^2(X)=1} \E[f(X)f(X')]\leq \rho_m(X;X').$$
Thus, $\rho^2_m(X;Y)=\rho_m(X;X')$ which completes the proof.
\end{proof}
Now a lower bound of $\hat{g}_{\eps}(X;Y)$ can be readily obtained.
\begin{corollary}\label{corollary_lowerbound_g_hat}
For a given joint distribution $P_{XY}$ defined over  $\X\times \Y$, we have for any $\eps>0$
$$\hat{g}_{\eps}(X;Y)\geq \frac{H(Y)}{\rho_m^2(X;Y)}\min\{\eps, \rho_m^2(X;Y)\}.$$
\end{corollary}
\begin{proof}
By Lemma~\ref{Lemma_Data_Processing_Maximal_corr}, we know that for any Markov chain $X\markov Y\markov Z$, we have $\rho_m(X;Z)\leq \rho_m(X;Y)$ and hence for $\eps\geq \rho^2_m(X;Y)$, the privacy constraint $\rho^2_m(X;Z)\leq \eps$ is not restrictive and hence $\hat{g}_{\eps}(X;Y)=H(Y)$ by setting $Y=Z$. For $0<\eps\leq\rho^2_m(X;Y)$, Lemma~\ref{non-increasing-Lemma_g_hat} implies that
$$\frac{\hat{g}_{\eps}(X;Y)}{\eps}\geq \frac{H(Y)}{\rho_m^2(X;Y)},$$
from which the result follows.
\end{proof}

A loose upper bound of $\hat{g}_{\eps}(X;Y)$ can be obtained using an argument similar to the one used for $g_{\eps}(X;Y)$. For the Markov chain $X\markov Y\markov Z$, we have
\begin{eqnarray}
  I(Y;Z) &=& I(X;Z)+I(Y;Z|X)\leq I(X;Z)+H(Y|X),\nonumber \\
   &\stackrel{(a)}{\leq} & \log\left(k\rho^2_m(X;Z)+1\right)+H(Y|X),  \label{upperbound_g_hat}
\end{eqnarray}
where  $k:=|\X|-1$ and $(a)$ comes from \cite[Proposition 1]{MaximalCorr_Secrecy}. We can, therefore, conclude from \eqref{upperbound_g_hat} and Corollary~\ref{corollary_lowerbound_g_hat} that
\begin{equation}\label{g_hat_LB_UN}
  \eps\frac{H(Y)}{\rho_m^2(X;Y)}\leq \hat{g}_{\eps}(X;Y)\leq \log\left(k\eps+1\right) + H(Y|X).
\end{equation}

Similar to Lemma~\ref{Lemma_Concavity}, the following lemma shows that the $\hat{g}_{\eps}(X;Y)$ is a concave function of $\eps$.
\begin{lemma}\label{Lemma_concavity_gHat}
For any given pair of random variables $(X,Y)$ with distribution $P$ over $\X\times \Y$, the mapping $\eps\mapsto \hat{g}_{\eps}(X;Y)$ is concave for $\eps\geq 0$.
\end{lemma}
\begin{proof}
The proof is similar to that of Lemma~\ref{Lemma_Concavity} except that here for two optimal filters $P_{Z_1|Y}:Y\to Z_1$ and $P_{Z_3|Y}:Y\to Z_3$ in $\hat{\D}_{\eps_1}(P)$ and $\hat{\D}_{\eps_3}(P)$, respectively, and  the random channel $P_{Z_{\lambda}|Y}:Y\to Z$ with output alphabet $\Z_1\cup \Z_3$ constructed using a coin flip with probability $\gamma$, we need to show that $P_{Z_{\lambda}|Y}\in\hat{\D}_{\eps_2}(P)$, where $0\leq \eps_1<\eps_2<\eps_3\leq \rho^2_m(X;Y)$. To show this, consider $f:\X\to \R$ such that $\E[f(X)]=0$ and $\E[f^2(X)]=1$ and let $U$ be a binary random variable as in the proof of Lemma~\ref{Lemma_Concavity}. We then have
\begin{eqnarray}
%% \nonumber to remove numbering (before each equation)
\E[\E^2[f(X)|Z_{\lambda}]] &=&\E\left[\E[\E^2[f(X)|Z_{\lambda}]|U]\right]\nonumber\\
&\stackrel{(a)}{=}& \gamma \E[\E^2[f(X)|Z_3]]+(1-\gamma) \E[\E^2[f(X)|Z_1]],\label{proof_Concavity_gHat}
%  \E[f(X)g(Z)] &=& \E_U[\E[f(X)g(Z)|U]]\stackrel{(a)}{=}\gamma\E[f(X)g(Z_3)]+(1-\gamma)\E[f(X)g(Z_1)] \nonumber\\
%   &=& \gamma \E[f(X)]\E[g(Z_3)]+\gamma\E[f(X)(g(Z_3)-\E[g(Z_3)])]\nonumber\\
%   && +(1-\gamma) \E[f(X)]\E[g(Z_1)]+(1-\gamma)\E[f(X)(g(Z_1)-\E[g(Z_1)])]\nonumber\\
%   &\stackrel{(b)}{=}& \gamma\E[f(X)(g(Z_3)-\E[g(Z_3)])]+(1-\gamma)\E[f(X)(g(Z_1)-\E[g(Z_1)])]\nonumber\\
%   &\stackrel{(c)}{\leq}& \gamma\sqrt{\var(g(Z_3))}\rho_m(X;Z_3)+(1-\gamma)\sqrt{\var(g(Z_1))}\rho_m(X;Z_1)\label{proof_Concavity_gHat}
\end{eqnarray}
where $(a)$ comes from the fact that $U$ is  independent of $X$. We can then conclude from \eqref{proof_Concavity_gHat} and the alternative characterization of maximal correlation  \eqref{Maximal_correlation_Equivalent} that
\begin{eqnarray*}
  \rho^2_m(X;Z_{\lambda})&=&\sup_{f:\E[f(X)]=0, \E[f^2(X)]=1}\E[\E^2[f(X)|Z_{\lambda}]]  \\
   &=& \sup_{f:\E[f(X)]=0, \E[f^2(X)]=1}\left[\gamma \E[\E^2[f(X)|Z_3]]+(1-\gamma) \E[\E^2[f(X)|Z_1]]\right]\\
   &\leq& \gamma \rho^2_m(X;Z_3)+(1-\gamma) \rho^2_m(X;Z_1)\leq \gamma \eps_3+(1-\gamma)\eps_1,
\end{eqnarray*}
from which we can conclude that $P_{Z_{\lambda}|Y}\in\hat{\D}_{\eps_2}(P)$.
\end{proof}

\subsection{Non-Trivial Filters For Perfect Privacy}

As it becomes clear later, requiring that $g_0(X;Y)=0$ is a useful assumption for the analysis of $g_{\eps}(X;Y)$. Thus, it is interesting to find a necessary and sufficient condition on the joint distribution $P_{XY}$ which results in $g_0(X;Y)=0$.
\begin{definition}[\cite{berger}] \label{definition_weakly_ind}
 The random variable $X$ is said to be \emph{weakly independent} of $Y$ if the rows of the transition matrix $P_{X|Y}$, i.e., the set of vectors $\{P_{X|Y}(\cdot|y), ~y\in\mathcal{Y}\}$, are linearly dependent.
 \end{definition}
The following lemma provides a necessary and sufficient condition for $g_0(X;Y)>0$.
%For example, the following channel for $p_1, p_2, p_3\geq 0$ and $p_1+p_2+p_3=1$
%\[ P_{X|Y}=\left( \begin{array}{ccc}
%p_1 & p_2 & p_3 \\
%p_3 & p_1 & p_2 \\
%p_2 & p_3 & p_2 \end{array} \right),\]
%is full rank for $p_1\neq p_2\neq p_3$ (because each row is a non-identical permutation of each other) and hence for this channel $X$ is not weakly independent of $Y$.
\begin{lemma}\label{Lemma_nece_suff}
For a given $(X,Y)$ with a given joint distribution $P_{XY}=P_Y\times P_{X|Y}$, $g_0(X;Y)>0$ (and equivalently $\hat{g}_0(X;Y)>0$) if and only if $X$ is weakly independent of $Y$.
\end{lemma}
\begin{proof}
$\Rightarrow$ direction:

Assuming that $g_0(X;Y)>0$ implies that there exists a random variable $Z$ over an alphabet $\Z$ such that the Markov condition $X\markov Y\markov Z$ is satisfied and $Z\indep X$ while $I(Y;Z)>0$. Hence, for any $z_1$ and $z_2$ in $\Z$, we must have $P_{X|Z}(x|z_1)=P_{X|Z}(x|z_2)$ for all $x\in \X$, which implies that
$$\sum_{y\in \Y} P_{X|Y}(x|y)P_{Y|Z}(y|z_1)=\sum_{y\in \Y} P_{X|Y}(x|y)P_{Y|Z}(y|z_2)$$ and hence $$\sum_{y\in \Y} P_{X|Y}(x|y)\left[P_{Y|Z}(y|z_1)-P_{Y|Z}(y|z_2)\right]=0.$$ Since $Y$ is not independent of $Z$, there exist $z_1$ and $z_2$ such that $P_{Y|Z}(y|z_1)\neq P_{Y|Z}(y|z_2)$ and hence the above shows that the set of vectors $P_{X|Y}(\cdot|y)$, $y\in \Y$ is linearly dependent.
\clearpage
$\Leftarrow$ direction:

Berger and Yeung \cite[Appendix II]{berger}, in a completely different context, showed that if $X$ being weakly independent of $Y$, one can always construct a binary random variable $Z$ correlated with $Y$ which satisfies $X\markov Y\markov Z$ and $X\indep Z$, and hence $g_0(X;Y)>0$.
\end{proof}
\begin{remark}
Lemma~\ref{Lemma_nece_suff} first appeared in \cite{Asoodeh_Allerton}. However,
Calmon et al.\ \cite{Calmon_fundamental-Limit} studied \eqref{Dual_gEpsilon}, the dual version of $g_{\eps}(X;Y)$, and showed an equivalent result for $t_{R}(X;Y)$. In fact, they showed that for a given $P_{XY}$, one can always generate $Z$ such that $I(X;Z)=0$,  $I(Y;Z)>0$ and $X\markov Y\markov Z$, or equivalently $g_0(X;Y)>0$, if and only if the smallest singular value of the conditional expectation operator $f\mapsto \E[f(X)|Y]$ is zero. This condition can, in fact, be shown to be equivalent to the condition in Lemma~\ref{Lemma_nece_suff}.

\end{remark}
\begin{remark}\label{Remark_on_Weak_Independence}
It is clear that, according to Definition~\ref{definition_weakly_ind}, $X$ is weakly independent of $Y$ if $|\Y|>|\X|$. Hence, Lemma~\ref{Lemma_nece_suff} implies that  $g_0(X;Y)>0$ if $Y$ has strictly larger alphabet than $X$.
\end{remark}
In light of the above remark, in the most common case of $|\Y|=|\X|$, one might have $g_0(X;Y)=0$, which corresponds to the most conservative scenario as no privacy leakage implies no broadcasting of observable data. In such cases, the rate of increase of $g_{\eps}(X;Y)$ at $\eps=0$, that is $g'_0(X;Y):=\frac{\text{d}}{\text{d}\eps}g_{\eps}(X;Y)|_{\eps=0}$, which corresponds to the initial efficiency of privacy-constrained information extraction,  proves to be very important in characterizing the behavior of $g_{\eps}(X;Y)$ for all $\eps\geq 0$. This is because, for example, by concavity of $\eps\mapsto g_{\eps}(X;Y)$, the slope of $g_{\eps}(X;Y)$ is maximized at $\eps=0$ and so
$$g'_0(X;Y)=\lim_{\eps\to 0}\frac{g_{\eps}(X;Y)}{\eps}=\sup_{\eps> 0}\frac{g_{\eps}(X;Y)}{\eps},$$ and hence  $g_{\eps}(X;Y)\leq \eps g'_0(X;Y)$ for all $\eps\leq I(X;Y)$ which, together with \eqref{UB&LB}, implies that $g_{\eps}(X;Y)=\eps\frac{H(Y)}{I(X;Y)}$ if $g'_0(X;Y)\leq\frac{H(Y)}{I(X;Y)}$. In the sequel, we always assume that $X$ is not weakly independent of $Y$, or equivalently $g_0(X;Y)=0$. For example, in light of Lemma~\ref{Lemma_nece_suff} and Remark~\ref{Remark_on_Weak_Independence}, we can assume that $|\Y|\leq |\X|$.

 It is easy to show that, $X$ is weakly independent of binary $Y$ if and only if $X$ and $Y$ are independent (see e.g., \cite[Remark 2]{berger}). The following corollary, therefore, immediately follows from Lemma~\ref{Lemma_nece_suff}.
 \begin{corollary}\label{generalizedtheorem1}
Let $Y$ be a non-degenerate binary random variable correlated with $X$. Then $g_0(X;Y)=0$.
\end{corollary}

\section{Operational Interpretations of the Rate-Privacy Function}

In this section, we provide a scenario in which $g_{\eps}(X;Y)$ appears as a boundary point of an achievable rate region and thus giving an information-theoretic operational interpretation for $g_{\eps}(X;Y)$. We then proceed to present an estimation-theoretic motivation for $\hat{g}_{\eps}(X;Y)$.

\subsection{Dependence Dilution}

Inspired by the problems of information amplification \cite{Cover_State_Amplification} and state masking \cite{Merhav_state_masking},  Courtade \cite{Courtade_Aplification} proposed the \emph{information-masking tradeoff} problem as follows. The tuple $(R_u, R_v, \Delta_A, \Delta_M)\in\R^4$ is said to be achievable if for two given separated sources $U\in\U$ and $V\in\V$ and any $\eps>0$ there exist mappings $f:\U^n\to \{1, 2, \dots, 2^{nR_u}\}$
 and $g:\V^n\to \{1, 2, \dots, 2^{nR_v}\}$ such that $I(U^n; f(U^n), g(V^n))\leq n(\Delta_M+\eps)$ and $I(V^n; f(U^n), g(V^n))\geq n(\Delta_A-\eps)$. In other words, $(R_u, R_v, \Delta_A, \Delta_M)$ is achievable if there exist indices $K$ and $J$ of rates $R_u$ and $R_v$ given $U^n$ and $V^n$, respectively, such that the receiver in possession of $(K, J)$ can recover at most $n\Delta_M$ bits about $U^n$ and at least $n\Delta_A$ about $V^n$. The closure of the set of all achievable tuple $(R_u, R_v, \Delta_A, \Delta_M)$ is characterized in \cite{Courtade_Aplification}. Here, we look at a similar problem but for a joint encoder. In fact, we want to examine the achievable rate of an encoder observing both $X^n$ and $Y^n$ which masks $X^n$ and amplifies $Y^n$ at the same time, by rates $\Delta_M$ and $\Delta_A$, respectively.

 We define a $(2^{nR}, n)$ \emph{dependence dilution} code by an encoder
$$f_n:\X^n\times \Y^n\to \{1,2,\dots, 2^{nR}\},$$
and a list decoder
$$g_n:\{1,2,\dots, 2^{nR}\}\to 2^{\Y^n},$$ where $2^{\Y^n}$ denotes the power set of $\Y^n$. A \emph{dependence dilution triple} $(R, \Delta_A, \Delta_M)\in\R^3_+$ is said to be achievable if, for any $\delta>0$, there exists a $(2^{nR}, n)$ dependence dilution code that for sufficiently large $n$ satisfies the utility constraint:
\begin{equation}\label{list_decoder}
    \Pr\left(Y^n\notin g_n(J)\right)<\delta
\end{equation}
having a fixed list size
\begin{equation}\label{lis_size}
    |g_n(J)|=2^{n(H(Y)-\Delta_A)},\qquad \forall J\in\{1,2,\dots, 2^{nR}\}
\end{equation}
where $J:=f_n(X^n, Y^n)$ is the encoder's output, and satisfies the privacy constraint:
\begin{equation}\label{privacy_constraint_list_decoder}
    \frac{1}{n}I(X^n; J)\leq \Delta_M+\delta.
\end{equation}
Intuitively speaking, upon receiving $J$, the decoder is required to construct list $g_n(J)\subset \Y^n$ of fixed size which contains likely candidates of the actual sequence $Y^n$. Without any observation, the decoder can only construct a list of size $2^{nH(Y)}$ which contains $Y^n$ with probability close to one. However, after $J$ is observed and the list $g_n(J)$ is formed, the decoder's list size can be reduced to $2^{n(H(Y)-\Delta_A)}$ and thus reducing the uncertainty about $Y^n$ by $0\leq n\Delta_A\leq nH(Y)$. This observation led Kim et al.\ \cite{Cover_State_Amplification} to show that the utility constraint \eqref{list_decoder} is equivalent to the amplification requirement
\begin{equation}\label{amplification_equivalence}
    \frac{1}{n}I(Y^n; J)\geq \Delta_A-\delta,
\end{equation}
which lower bounds the amount of information $J$ carries about $Y^n$. The following lemma gives an outer bound for the achievable dependence dilution region.
\begin{theorem}\label{Theorem_correlation_dilution}
Any achievable dependence dilution triple $(R, \Delta_A, \Delta_M)$ satisfies
%\begin{eqnarray*}
%% \nonumber to remove numbering (before each equation)
%  R &\geq & \Delta_A \\
%  \Delta_A &\leq & I(Y;U) \\
%   \Delta_M &\geq & I(X;U)-I(Y;U)+\Delta_A,
%\end{eqnarray*}
 \begin{equation*}
  \begin{cases}
        ~~\,R \geq \Delta_A\\
        \Delta_A \leq I(Y;U)\\
        \Delta_M \geq I(X;U)-I(Y;U)+\Delta_A,
        \end{cases}
 \end{equation*}
for some  auxiliary random variable $U\in\U$ with a finite alphabet and jointly distributed with $X$ and $Y$.
\end{theorem}
Before we prove this theorem, we need two preliminary lemmas. The first lemma is an extension of Fano's inequality for list decoders and the second one makes use of a single-letterization technique to express $I(X^n; J)-I(Y^n; J)$ in a single-letter form in the sense of Csisz\'{a}r and K\"{o}rner \cite{csiszarbook}.
\begin{lemma}[\cite{Cover_State_Amplification,Ahlswede_korner_lossless_decoder}]\label{lemma_fano_list}
Given a pair of random variables $(U,V)$ defined over $\U\times \V$ for finite $\V$ and arbitrary $\U$, any list decoder $g: \U\to 2^{\V}$, $U\mapsto g(U)$ of fixed list size $m$ (i.e., $|g(u)|=m, ~\forall u\in \U$), satisfies
$$H(V|U)\leq h_b(p_e)+p_e\log |\V|+(1-p_e)\log m,$$
where $p_e:=\Pr(V\notin g(U))$ and $h_b:[0, 1]\to [0, 1]$ is the binary entropy function.
\end{lemma}
%\begin{lemma}[\cite{Fady_Lecture_note}]\label{lemma_fano_list}
%Let $U$ and $V$ be two random variables with alphabets $\U$ and $\V$, respectively, where $\U$ is finite and $\V$ can be countably many. For any list decoder $g:\V\to 2^{\U}$ with fixed list size $m$, we have
%$$H(U|V)\leq h(p_e^m)+p_e^m\log(|\U|-\psi)+(1-p_e^m)\log \psi,$$
%where $p_e^m:=\Pr(U\notin g(V))$ and $$\psi:=\sum_{u\in \U}\sum_{A\subset \X^m:u\in A}P_V(g^{-1}(A)).$$
%\end{lemma}
This lemma, applied to $J$ and $Y^n$ in place of $U$ and $V$, respectively, implies that for any list decoder with the property \eqref{list_decoder}, we have
\begin{equation}\label{Fano_list_decoder}
    H(Y^n|J)\leq \log |g_n(J)|+n\eps_n,
\end{equation}
where $\eps_n:=\frac{1}{n}+(\log |\Y|-\frac{1}{n}\log |g_n(J)|)p_e$ and hence
$\eps_n\to 0$ as $n\to \infty$.

\begin{lemma}\label{lemma_single_letterization}
Let $(X^n, Y^n)$ be $n$ i.i.d. copies of a pair of random variables $(X,Y)$. Then for a random variable $J$ jointly distributed with $(X^n, Y^n)$, we have
$$I(X^n;J)-I(Y^n;J)=\sum_{i=1}^n[I(X_i; U_i)-I(Y_i; U_i)],$$ where $U_i:=(J, X_{i+1}^n, Y^{i-1})$.
\end{lemma}
\begin{proof}
Using the chain rule for the mutual information, we can express $I(X^n; J)$ as follows
\begin{eqnarray}
  I(X^n; J) &=& \sum_{i=1}^n I(X_i; J|X_{i+1}^n)=\sum_{i=1}^n I(X_i; J,X_{i+1}^n)\nonumber\\
   &=& \sum_{i=1}^n[I(X_i; J,X_{i+1}^n, Y^{i-1})-I(X_i; Y^{i-1}|J, X_{i+1}^n)] \nonumber\\
   &=&  \sum_{i=1}^nI(X_i; U_i)-\sum_{i=1}^nI(X_i; Y^{i-1}|J, X_{i+1}^n).\label{proof_single_letterizaton1}
\end{eqnarray}
Similarly, we can expand $I(Y^n; J)$ as
\begin{eqnarray}
  I(Y^n; J) &=& \sum_{i=1}^n I(Y_i; J|Y^{i-1})=\sum_{i=1}^n I(Y_i; J,Y^{i-1})\nonumber\\
   &=& \sum_{i=1}^n[I(Y_i; J,X_{i+1}^n, Y^{i-1})-I(Y_i; X_{i+1}^n|J, Y^{i-1})] \nonumber\\
   &=&  \sum_{i=1}^nI(Y_i; U_i)-\sum_{i=1}^nI(Y_i; X_{i+1}^n|J, Y^{i-1}).\label{proof_single_letterizaton2}
\end{eqnarray}
Subtracting \eqref{proof_single_letterizaton2} from \eqref{proof_single_letterizaton1}, we get
\begin{eqnarray*}
  I(X^n; J)-I(Y^n; J) &=& \sum_{i=1}^n[I(X_i; U_i)-I(Y_i; U_i)]-\sum_{i=1}^n[I(X_i; Y^{i-1}|J, X_{i+1}^n)-I(X_{i+1}^n ;Y_i|J, Y^{i-1})] \nonumber \\
   &\stackrel{(a)}{=}&  \sum_{i=1}^n[I(X_i; U_i)-I(Y_i; U_i)],
\end{eqnarray*}
where $(a)$ follows from the Csisz\'{a}r sum identity \cite{networkinfotheory}.
\end{proof}
\begin{proof}[Proof of Theorem~\ref{Theorem_correlation_dilution}]
The rate $R$ can be bounded as
\begin{eqnarray}
  nR &\geq& H(J)\geq I(Y^n;J) \\
   &=& nH(Y)-H(Y^n|J) \nonumber\\
   &\stackrel{(a)}{\geq} & nH(Y)-\log |g_n(J)|-n\eps_n\nonumber\\
   &\stackrel{(b)}{=}& n\Delta_A-n\eps_n,  \label{proof_theorem_rate}
\end{eqnarray}
where $(a)$ follows from Fano's inequality \eqref{Fano_list_decoder} with $\eps_n\to 0$ as $n\to \infty$ and $(b)$ is due to \eqref{lis_size}.
We can also upper bound $\Delta_A$ as
\begin{eqnarray}
  \Delta_A &\stackrel{(a)}{=}& H(Y^n)-\log |g_n(J)| \nonumber\\
   &\stackrel{(b)}{\leq} &  H(Y^n)-H(Y^n|J)+n\eps_n\nonumber\\
  &=& \sum_{i=1}^n H(Y_i)-H(Y_i|Y^{i-1}, J)+n\eps_n\nonumber\\
&\leq &\sum_{i=1}^n H(Y_i)-H(Y_i|Y^{i-1}, X_{i+1}^n, J)+n\eps_n\nonumber\\
&=&\sum_{i=1}^nI(Y_i; U_i)+n\eps_n, \label{proof_theorem_amplification}
\end{eqnarray}
where $(a)$ follows from \eqref{lis_size}, $(b)$ follows from \eqref{Fano_list_decoder}, and in the last equality the auxiliary random variable $U_i:=(Y^{i-1}, X_{i+1}^n, J)$ is introduced.

We shall now lower bound $I(X^n; J)$:
%\hspace{-0.7cm}\begin{eqnarray*}
%% \nonumber to remove numbering (before each equation)
%  H(X^n|J) &\stackrel{(a)}{=}& H(Y^n|J)+\sum_{i=1}^n[H(X|U_i)-H(Y_i|U_i)] \\
%   &\stackrel{(b)}{\leq}& n\eps_n +H(Y^n)-n\Delta_A\\
%   &&+\sum_{i=1}^n[H(X_i|U_i)-H(Y_i|U_i)] \\
%   &=& n\eps_n-n\Delta_A+\sum_{i=1}^n [H(X_i|U_i)+I(Y_i; U_i)]
%\end{eqnarray*}
\begin{eqnarray}
    \hspace{-0.5cm}n(\Delta_M+\delta)&\geq& I(X^n; J)\nonumber\\
    &\stackrel{(a)}{=}&  I(Y^n; J)+\sum_{i=1}^n [I(X_i;U_i)-I(Y_i; U_i)]\nonumber\\
 %   &=&  H(Y^n)-H(Y^n|J)+\sum_{i=1}^n [I(X_i;U_i)-I(Y_i; U_i)]\nonumber\\
    &\stackrel{(b)}{\geq}& n\Delta_A +\sum_{i=1}^n [I(X_i;U_i)-I(Y_i; U_i)]-n\eps_n\label{proof_theorem_masking}.
\end{eqnarray}
where $(a)$ follows from Lemma~\ref{lemma_single_letterization} and $(b)$ is due to Fano's inequality and \eqref{lis_size} (or equivalently from \eqref{amplification_equivalence}).

Combining \eqref{proof_theorem_rate}, \eqref{proof_theorem_amplification} and \eqref{proof_theorem_masking}, we can write
 \begin{eqnarray*}
   R &\geq & \Delta_A-\eps_n \\
   \Delta_A &\leq& I(Y_Q; U_Q|Q)+\eps_n= I(Y_Q; U_Q,Q)+\eps_n\\
   \Delta_M&\geq& \Delta_A +I(X_Q;U_Q|Q)-I(Y_Q; U_Q|Q)-\eps'_n\\
   &=&\Delta_A +I(X_Q;U_Q,Q)-I(Y_Q; U_Q,Q)-\eps'_n
 \end{eqnarray*}
 where $\eps'_n:=\eps_n+\delta$ and $Q$ is a random variable distributed uniformly over $\{1,2,\dots, n\}$ which is independent of $(X,Y)$ and hence $I(Y_Q; U_Q|Q)=\frac{1}{n}\sum_{i=1}^nI(Y_i; U_i)$. The results follow by denoting $U:=(U_Q,Q)$ and noting that $Y_Q$ and $X_Q$ have the same distributions as $Y$ and $X$, respectively.
\end{proof}
If the encoder does not have direct access to the private source $X^n$, then we can define the encoder mapping as $f_n:\Y^n\to \{1, 2, \dots, s^{nR}\}$. The following corollary is an immediate consequence of Theorem~\ref{Theorem_correlation_dilution}.
\begin{corollary}\label{Corollary_Dilution}
If the encoder does not see the private source, then for all achievable dependence dilution triple $(R, \Delta_A, \Delta_M)$, we have
%\begin{eqnarray*}
%% \nonumber to remove numbering (before each equation)
%  R &\geq & \Delta_A \\
%  \Delta_A &\leq & I(Y;U) \\
%   \Delta_M &\geq & I(X;U)-I(Y;U)+\Delta_A,
%\end{eqnarray*}
 \begin{equation*}
  \begin{cases}
        ~~\,R \geq  \Delta_A \\
        \Delta_A \leq I(Y;U)\\
        \Delta_M \geq I(X;U)-I(Y;U)+\Delta_A,
        \end{cases}
 \end{equation*}
for some joint distribution $P_{XYU}=P_{XY}P_{U|Y}$ where the auxiliary random variable $U\in\U$ satisfies $|\U|\leq |\Y|+1$.
\end{corollary}
%\begin{remark}
%Setting $U=Y$ in Corollary~\ref{Corollary_Dilution}, we can recover the amplification-masking region obtained by Courtade in \cite{Courtade_Aplification}.
%\end{remark}
\begin{remark}
If source $Y$ is required to be amplified (according to \eqref{amplification_equivalence}) at maximum rate, that is, $\Delta_A=I(Y; U)$ for an auxiliary random variable $U$ which satisfies $X\markov Y\markov U$, then by Corollary~\ref{Corollary_Dilution}, the best privacy performance one can expect from the dependence dilution setting is
\begin{equation}\label{Dual_gEpsilon_2}
  \Delta_M^*=\min_{\substack{U:X~\markov ~~Y~\markov ~~U\\I(Y; U)\geq \Delta_A}} I(X;U),
\end{equation}
which is equal to the dual of $g_{\eps}(X;Y)$ evaluated at $\Delta_A$, $t_{\Delta_A}(X;Y)$, as defined in \eqref{Dual_gEpsilon}.
\end{remark}
The dependence dilution problem is closely related to the discriminatory lossy source coding problem studied in \cite{Tandon_Quantize_bin}. In this problem, an encoder $f$ observes $(X^n, Y^n)$ and wants to describe this source to a decoder, $g$, such that $g$ recovers $Y^n$ within distortion level $D$ and $I(f(X^n, Y^n); X^n)\leq n\Delta_M$. If the distortion level is Hamming measure, then the distortion constraint and the amplification constraint are closely related via Fano's inequality. Moreover, dependence dilution problem reduces to a secure lossless (list decoder of fixed size 1) source coding problem by  setting $\Delta_A=H(H)$, which is recently studied in \cite{Asoodeh_Allerton2015}.
\subsection{MMSE Estimation of Functions of Private Information}

In this section, we provide a justification for the privacy guarantee $\rho_m^2(X;Z)\leq \eps$. To this end, we recall the definition of the minimum mean squared error estimation.

\begin{definition}
Given random variables $U$ and $V$, $\mmse(U|V)$ is defined as the minimum error of an estimate, $g(V)$, of $U$ based on $V$, measured in the mean-square sense, that is
\begin{equation}\label{MMSE_Definition}
  \mmse(U|V):=\inf_{g\in\L^2(\V)}\E[\left(U-g(V)\right)^2]=\E[\left(U-\E[U|V]\right)^2]=\E[\var(U|V)],
\end{equation}
where $\var(U|V)$ denotes the conditional variance of $U$ given $V$.
\end{definition}
It is easy to see that $\mmse(U|V)=0$ if and only if $U=f(V)$ for some measurable function $f$ and $\mmse(U|V)=\var(U)$ if and only if $U\indep V$. Hence, unlike for the case of maximal correlation, a small value of $\mmse(U|V)$ implies a strong dependence between $U$ and $V$. Hence, although it is not a "proper" measure of correlation, in a certain sense it measures how well one random variable can be predicted from another one.

Given a non-degenerate measurable function $f:\X\to \R$, consider the following constraint on $\mmse(f(X)|Y)$
\begin{equation}\label{Privacy_constraint_MMSE}
(1-\eps)\var(f(X))\leq \mmse(f(X)|Z)\leq \var(f(X)).
\end{equation}
   This guarantees that no adversary knowing $Z$ can efficiently estimate $f(X)$. First consider the case where $f$ is an identity function, i.e., $f(x)=x$. In this case, a direct calculation shows that
 \begin{eqnarray*}
  \mmse(X|Z)&\stackrel{(a)}{=}& \E[(X-\E[X|Z])^2]=\E[X^2]-\E[(\E[X|Z])^2]\\
    &=& \var(X)(1-\rho^2(X;\E[X|Z])) \\
    &\stackrel{(b)}{\geq} & \var(X)(1-\rho_m^2(X;Z)),
 \end{eqnarray*}
where $(a)$ follows from \eqref{MMSE_Definition} and $(b)$ is due to the definition of maximal correlation. Having imposed $\rho_m^2(X;Z)\leq \eps$, we, can therefore conclude that the MMSE of estimating $X$ given $Z$ satisfies
\begin{equation}\label{Privacy_constraint_MMSE_Identity}
(1-\eps)\var(X)\leq \mmse(X|Z)\leq \var(X),
\end{equation}
which shows that $\rho_m^2(X;Z)\leq \eps$ implies \eqref{Privacy_constraint_MMSE} for $f(x)=x$. However, in the following we show that the constraint $\rho_m^2(X;Z)\leq \eps$ is, indeed, equivalent to \eqref{Privacy_constraint_MMSE} for \emph{any} non-degenerate measurable $f:\X\to \R$.
\begin{definition}[\cite{Maxim_Poincare_inequality}]\label{Poincare_Definition}
A joint distribution $P_{UV}$ satisfies a \emph{Poincar\'{e} inequality} with constant $c\leq 1$ if for all $f:\U\to \R$
$$c\cdot\var(f(U))\leq  \mmse(f(U)|V),$$ and the \emph{Poincar\'{e} constant} for $P_{UV}$ is defined as $$\vartheta(U;V):=\inf_{f}\frac{\mmse(f(U)|V)}{\var(f(U))}.$$
\end{definition}

The privacy constraint \eqref{Privacy_constraint_MMSE} can then be viewed as
\begin{equation}\label{privacy_MMSE_Poincare}
  \vartheta(X;Z)\geq 1-\eps.
\end{equation}
\begin{theorem}[\cite{Maxim_Poincare_inequality}]\label{Theorem_Poincare}
For any joint distribution $P_{UV}$, we have
$$\vartheta(U;V)=1-\rho_m^2(U;V).$$
\end{theorem}
In light of Theorem~\ref{Theorem_Poincare} and \eqref{privacy_MMSE_Poincare}, the privacy constraint \eqref{Privacy_constraint_MMSE} is equivalent to $\rho_m^2(X;Z)\leq \eps$, that is,
$$\rho_m^2(X;Z)\leq \eps\Longleftrightarrow (1-\eps)\var(f(X))\leq \mmse(f(X)|Z)\leq \var(f(X)),$$
for any non-degenerate measurable functions $f:\X\to \R$.

 Hence, $\hat{g}_{\eps}(X;Y)$ characterizes the maximum information extraction from $Y$ such that no (non-trivial) function of $X$ can be efficiently estimated, in terms of MMSE \eqref{Privacy_constraint_MMSE}, given the extracted information.

\section{Observation Channels for Minimal and Maximal $ {\bf g_{\eps}(X;Y)}$}

In this section, we characterize the observation channels which achieve the lower or upper bounds on the rate-privacy function in \eqref{UB&LB}. We first derive general conditions for achieving the lower bound and then present a large family of observation channels $P_{Y|X}$ which achieve the lower bound. We also give a family of $P_{Y|X}$ which attain the upper bound on $g_{\eps}(X;Y)$.
\subsection{Conditions for Minimal $g_{\eps}(X;Y)$}

Assuming that $g_0(X;Y)=0$, we seek a set of conditions on $P_{XY}$ such that $g_{\eps}(X;Y)$ is linear in $\eps$, or equivalently, $g_{\eps}(X;Y)=\eps\frac{H(Y)}{I(X;Y)}$. In order to do this, we shall examine the slope of $g_{\eps}(X;Y)$ at zero. Recall that by concavity of $g_{\eps}(X;Y)$, it is clear that $g'_0(X;Y)\geq \frac{H(Y)}{I(X;Y)}$. We strengthen this bound in the following lemmas. For this, we need to recall the notion of Kullback-Leibler divergence. Given two probability distribution $P$ and $Q$ supported over a finite alphabet $\U$,
\begin{equation}\label{KL_Divergence}
  D(P||Q):=\sum_{u\in\U}P(u)\log\left(\frac{P(u)}{Q(u)}\right).
\end{equation}
\begin{lemma}\label{Lemma_slope_lowerbound}
For a given joint distribution $P_{XY}=P_Y\times P_{X|Y}$, if $g_0(X;Y)=0$, then for any $\eps\geq0$ $$g'_{0}(X;Y)\geq \max_{y\in \Y}\frac{-\log P_Y(y)}{D(P_{X|Y}(\cdot|y)||P_X(\cdot))}.$$
\end{lemma}
\begin{proof}
The proof is given in Appendix~\ref{AppendixI_Theorem_gHat_lowerbound}.
\end{proof}
\begin{remark}
Note that if for a given joint distribution $P_{XY}$, there exists $y_0\in\Y$ such that $D(P_{X|Y}(\cdot|y_0)||P_X(\cdot))=0$, it implies that $P_{X|Y}(\cdot|y_0)=P_X(x)$. Consider the binary random variable $Z\in\{1, \text{e}\}$ constructed according to the distribution $P_{Z|Y}(1|y_0)=1$ and $P_{Z|Y}(\text{e}|y)=1$ for $y\in\Y\textbackslash\{y_0\}$. We can now claim that $Z$ is independent of $X$, because $P_{X|Z}(x|1)=P_{X|Y}(x|y_0)=P_X(x)$ and
\begin{eqnarray*}
  P_{X|Z}(x|\text{e}) &=& \sum_{y\neq y_0}P_{X|Y}(x|y)P_{Y|Z}(y|\text{e})=\sum_{y\neq y_0}P_{X|Y}(x|y)\frac{P_Y(y)}{1-P_Y(y_0)} \\
   &=& \frac{1}{1-P_Y(y_0)}\sum_{y\neq y_0}P_{XY}(x,y)=P_X(x).
\end{eqnarray*}
Clearly, $Z$ and $Y$ are not independent, and hence $g_0(X;Y)>0$. This implies that the right-hand side of inequality in Lemma~\ref{Lemma_slope_lowerbound} can not be infinity.
\end{remark}

In order to prove the main result, we need the following simple lemma.
\begin{lemma}\label{Lemma_slope_upperbound}
For any joint distribution $P_{XY}$, we have
$$\frac{H(Y)}{I(X;Y)}\leq \max_{y\in\Y}\frac{-\log P_Y(y)}{D(P_{X|Y}(\cdot|y)||P_X(x))},$$ where equality holds if and only if there exists a constant $c>0$ such that $-\log P_Y(y)=cD(P_{X|Y}(\cdot|y)||P_X(x))$ for all $y\in\Y$.
\end{lemma}
\begin{proof}
It is clear that $$\frac{H(Y)}{I(X;Y)}=\frac{-\sum_{y\in\Y}P_Y(y)\log P_Y(y)}{\sum_{y\in\Y}P_Y(y)D(P_{X|Y}(\cdot|y)||P_X(x))}\leq \max_{y\in\Y}\frac{-\log P_Y(y)}{D(P_{X|Y}(\cdot|y)||P_X(x))},$$
where the inequality follows from the fact that for any three sequences of positive numbers $\{a_i\}_{i=1}^n$, $\{b_i\}_{i=1}^n$ and $\{\lambda_i\}_{i=1}^n$ we have $\frac{\sum_{i=1}^n\lambda_i a_i}{\sum_{i=1}^n\lambda_i b_i}\leq \max_{1\leq i\leq n}\frac{a_i}{b_i}$ where equality occurs if and only if $\frac{a_i}{b_i}=c$ for all $1\leq i\leq n$.
\end{proof}

Now we are ready to state the main result of this subsection.
\begin{theorem}\label{Theorem_slope-equality}
For a given $(X,Y)$ with joint distribution $P_{XY}=P_Y\times P_{X|Y}$, if $g_0(X;Y)=0$ and $\eps\mapsto g_{\eps}(X;Y)$ is linear for $0\leq\eps\leq I(X;Y)$, then for any $y\in\Y$
$$\frac{H(Y)}{I(X;Y)}= \frac{-\log P_Y(y)}{D(P_{X|Y}(\cdot|y)||P_X(\cdot))}.$$
\end{theorem}
\begin{proof}
Note that the fact that $g_0(X;Y)=0$ and $g_{\eps}(X;Y)$ is linear in $\eps$ is equivalent to  $g_{\eps}(X;Y)=\eps\frac{H(Y)}{I(X;Y)}$.  It is, therefore, immediate from Lemmas~\ref{Lemma_slope_lowerbound} and \ref{Lemma_slope_upperbound} that we have
\begin{eqnarray}
    g'_0(X;Y)&\stackrel{(a)}{=}&\frac{H(Y)}{I(X;Y)}\stackrel{(b)}{\leq} \max_{y\in\Y}\frac{-\log P_Y(y)}{D(P_{X|Y}(\cdot|y)||P_X(x))}\nonumber\\
&\stackrel{(c)}{\leq}&  g'_0(X;Y),\label{Proof_theorem}
\end{eqnarray}
where $(a)$ follows from the fact that $g_{\eps}(X;Y)=\eps\frac{H(Y)}{I(X;Y)}$ and $(b)$ and $(c)$ are due to Lemmas~\ref{Lemma_slope_upperbound} and~\ref{Lemma_slope_lowerbound}, respectively. The inequality in \eqref{Proof_theorem} shows that
\begin{equation}\label{Equality_in_Ratio}
  \frac{H(Y)}{I(X;Y)}=\max_{y\in\Y}\frac{-\log P_Y(y)}{D(P_{X|Y}(\cdot|y)||P_X(x))}.
\end{equation}
According to Lemma~\ref{Lemma_slope_upperbound}, \eqref{Equality_in_Ratio} implies that the ratio of $\frac{-\log P_Y(y)}{D(P_{X|Y}(\cdot|y)||P_X(x))}$ does not depend on $y\in\Y$ and hence the result follows.
\end{proof}

This theorem implies that if there exists $y=y_1$ and $y=y_2$ such that $\frac{\log P_Y(y)}{D(P_{X|Y}(\cdot|y)||P_X(x))}$ results in two different values, then $\eps\mapsto g_{\eps}(X,Y)$ cannot achieve the lower bound in \eqref{UB&LB}, or equivalently $$g_{\eps}(X;Y)>\eps\frac{H(Y)}{I(X;Y)}.$$ This, therefore, gives a necessary condition for the lower bound to be achievable. The following corollary simplifies this necessary condition.
\begin{corollary}\label{Corollary_Y_Uniform}
For a given joint distribution $P_{XY}=P_Y\times P_{X|Y}$, if $g_0(X;Y)=0$ and $\eps\mapsto g_{\eps}(X;Y)$ is linear, then the following are equivalent:
\begin{itemize}
\item[(i)] $Y$ is uniformly distributed,
\item[(ii)] $D(P_{X|Y}(\cdot|y)||P_X(\cdot))$ is constant for all $y\in\Y$.
\end{itemize}
\end{corollary}
\begin{proof}
$(i)\Rightarrow (ii)$:

From Theorem~\ref{Theorem_slope-equality}, we have for all $y\in\Y$
\begin{equation}\label{equation1}
    \frac{H(Y)}{I(X;Y)}= \frac{-\log(P_Y(y))}{D\left(P_{X|Y}(\cdot|y)||P_X(\cdot)\right)}.
\end{equation}
Letting $D:=D\left(P_{X|Y}(\cdot|y)||P_X(\cdot)\right)$ for any $y\in\Y$, we have $\sum_yP_{Y}(y)D=I(X;Y)$ and hence $D=I(X;Y)$, which together with \eqref{equation1} implies that $H(Y)= -\log(P_Y(y))$ for all $y\in \Y$ and hence $Y$ is uniformly distributed.

$(ii)\Rightarrow (i)$:

When $Y$ is uniformly distributed, we have from  \eqref{equation1} that $I(X;Y)=D\left(P_{X|Y}(\cdot|y)||P_X(\cdot)\right)$ which implies that $D\left(P_{X|Y}(\cdot|y)||P_X(\cdot)\right)$ is constant for all $y\in\Y$.
\end{proof}
\begin{example}
Suppose $P_{Y|X}$ is a binary symmetric channel (BSC) with crossover probability $0<\alpha<1$ and $P_X=\sBer(0.5)$. In this case, $P_{X|Y}$ is also a BSC with input distribution $P_{Y}=\sBer(0.5)$. Note that Corollary~\ref{generalizedtheorem1} implies that $g_0(X;Y)=0$. We will show that $g_{\eps}(X;Y)$ is linear as a function of $\eps\geq 0$ for a larger family of symmetric channels (including BSC) in Corollary~\ref{corollary_BISO_Uniform}. Hence, the BSC with uniform input nicely illustrates Corollary~\ref{Corollary_Y_Uniform}, because $D(P_{X|Y}(\cdot|y)||P_X(\cdot))=1-h(\alpha)$ for $y\in \{0,1\}$.
\end{example}
\begin{example}
Now suppose $P_{X|Y}$ is a binary asymmetric channel such that $P_{X|Y}(\cdot|0)=\sBer(\alpha_0)$, $P_{X|Y}(\cdot|1)=\sBer(\alpha_1)$ for some $0<\alpha_0, \alpha_1<1$ and input distribution $P_Y=\sBer(p)$, $0<p\leq 0.5$. It is easy to see that if $\alpha_0+\alpha_1=1$ then $D(P_{X|Y}(\cdot|y)||P_X(\cdot))$ does not depend on $y$ and hence we can conclude from Corollary~\ref{Corollary_Y_Uniform} (noticing that $g_0(X;Y)=0$) that in this case for any $p<0.5$, $g_{\eps}(X;Y)$ is not linear and hence for $0<\eps<I(X;Y)$
$$g_{\eps}(X;Y)>\eps\frac{H(Y)}{I(X;Y)}.$$
\end{example}

In Theorem~\ref{Theorem_slope-equality}, we showed that when $g_{\eps}(X;Y)$
achieves its lower bound, illustrated in \eqref{UB&LB}, the slope of the mapping $\eps\mapsto g_{\eps}(X;Y)$ at zero is equal to $\frac{-\log P_Y(y)}{D(P_{X|Y}(\cdot|y)||P_X(\cdot))}$ for any $y\in \Y$. We will show in the next section that the reverse direction is also true at least for a large family of binary-input symmetric output channels, for instance when $P_{Y|X}$ is a BSC, and thus showing that in this case, $$g'_0(X;Y)=\frac{-\log P_Y(y)}{D(P_{X|Y}(\cdot|y)||P_X(\cdot))}, ~~\forall y\in\Y\Longleftrightarrow g_{\eps}(X;Y)=\eps\frac{H(Y)}{I(X;Y)},~~~0\leq\eps\leq I(X;Y).$$

%\begin{conjecture}\label{Conjecture}
%For a given $(X,Y)$ with joint distribution $P_{XY}=P_Y\times P_{X|Y}$, then $g_{\eps}(X;Y)=\eps\frac{H(Y)}{I(X;Y)}$ if and only if $$\frac{H(Y)}{I(X;Y)}=\max_{y\in\Y} \frac{-\log P_Y(y)}{D(P_{X|Y}(\cdot|y)||P_X(\cdot))}.$$
%\end{conjecture}
%It is straightforward from \eqref{Lemma_combined_two_lemmas} that in order to prove the above conjecture it suffices to show that
%$$g'_0(X;Y)=\max_{y\in\Y} \frac{-\log P_Y(y)}{D(P_{X|Y}(\cdot|y)||P_X(\cdot))}.$$

\subsection{Special Observation Channels}

In this section, we apply the results of last section to different joint distributions $P_{XY}$. In the first family of channels from $X$ to $Y$, we look at the case where $Y$ is binary and the reverse channel $P_{X|Y}$ has symmetry in a particular sense, which will be specified later. One particular case of this family of channels is when $P_{X|Y}$ is a BSC. As a family of observation channels which achieves the upper bound of $g_{\eps}(X;Y)$, stated in \eqref{UB&LB}, we look at the class of erasure channels from $X\to Y$, i.e., $Y$ is an erasure version of $X$.

\subsubsection{Observation Channels With Symmetric Reverse}

The first example of $P_{XY}$ that we consider for binary $Y$ is the so-called \emph{Binary Input Symmetric Output} (BISO) $P_{X|Y}$, see for example \cite{Shamai_BISO1, Shamai_BISO2}. Suppose $\Y=\{0,1\}$ and $\X=\{0, \pm 1, \pm 2, \dots, \pm k\}$, and for any $x\in \X$ we have $P_{X|Y}(x|1)=P_{X|Y}(-x|0)$. This clearly implies that $p_0:=P_{X|Y}(0|0)=P_{X|Y}(0|1)$. We notice that with this definition of symmetry, we can always assume that the output alphabet $\X=\{\pm 1, \pm 2, \dots, \pm k\}$ has even number of elements because we can split $X=0$ into two outputs, $X=0^+$ and  $X=0^-$, with $P_{X|Y}(0^-|0)=P_{X|Y}(0^+|0)=\frac{p_0}{2}$ and $P_{X|Y}(0^-|1)=P_{X|Y}(0^+|1)=\frac{p_0}{2}$. The new channel is clearly essentially equivalent to the original one, see \cite{Shamai_BISO2} for more details.  This family of channels can also be characterized using the definition of \emph{quasi-symmetric} channels \cite[Definition 4.17]{Fady_Lecture_note}. A channel $\sW$ is BISO if (after making $|\X|$ even) the transition matrix $P_{X|Y}$ can be  partitioned along its columns into binary-input binary-output sub-arrays in which rows are permutations of each other and the column sums are equal. It is clear that binary symmetric channels and binary erasure channels are both BISO.
The following lemma gives an upper bound for $g_{\eps}(X,Y)$ when $P_{X|Y}$ belongs to such a family of channels.
\begin{lemma}\label{lemma_BISO}
If the channel $P_{X|Y}$ is BISO, then for $\eps\in [0, I(X;Y)]$,
$$\eps\frac{H(Y)}{I(X;Y)}\leq g_{\eps}(X;Y)\leq H(Y)-\frac{I(X;Y)-\eps}{C(P_{X|Y})},$$
where $C(P_{X|Y})$ denotes the capacity of $P_{X|Y}$.
\end{lemma}
\begin{proof}
The lower bound has already appeared in \eqref{UB&LB}. To prove the upper bound note that by Markovity $X\markov Y\markov Z$, we have for any $x\in \X$ and $z\in \Z$
\begin{equation}\label{Proof_BISO_1}
P_{X|Z}(x|z)=P_{X|Y}(x|0)P_{Y|Z}(0|z)+P_{X|Y}(x|1)P_{Y|Z}(1|z).
\end{equation}
Now suppose $\Z_0:=\{z: P_{Y|Z}(0|z)\leq P_{Y|Z}(1|z)\}$ and similarly $\Z_1:=\{z: P_{Y|Z}(1|z)\leq P_{Y|Z}(0|z)\}$. Then \eqref{Proof_BISO_1} allows us to write for $z\in \Z_0$
\begin{equation}\label{Proof_BISO_2}
P_{X|Z}(x|z)=P_{X|Y}(x|0)h_b^{-1}(H(Y|Z=z))+P_{X|Y}(x|1)(1-h_b^{-1}(H(Y|Z=z))),
\end{equation}
where $h_b^{-1}:[0, 1]\to [0, 0.5]$ is the inverse of binary entropy function, and for $z\in \Z_1$,
\begin{equation}\label{Proof_BISO_3}
P_{X|Z}(x|z)=P_{X|Y}(x|0)(1-h_b^{-1}(H(Y|Z=z)))+P_{X|Y}(x|1)h_b^{-1}(H(Y|Z=z)).
\end{equation}
Letting $P\otimes h_b^{-1}(H(Y|z))$ and $\tilde{P}\otimes h_b^{-1}(H(Y|z))$  denote the right-hand sides of \eqref{Proof_BISO_2} and \eqref{Proof_BISO_3}, respectively,  we can, hence, write
\begin{eqnarray*}
  H(X|Z) &=& \sum_{z\in\Z}P_Z(z)H(X|Z=z)\\
  &\stackrel{(a)}{=}&\sum_{z\in\Z_0}P_Z(z)H(P\otimes h_b^{-1}(H(Y|Z=z)))+\sum_{z\in\Z_1}P_Z(z)H(\tilde{P}\otimes h_b^{-1}(H(Y|Z=z))) \\
   &\stackrel{(b)}{\leq}&  \sum_{z\in\Z_0}P_Z(z)\left[(1-H(Y|Z=z))H(P\otimes h_b^{-1}(0))+H(Y|Z=z)H(P\otimes h_b^{-1}(1))\right]\\
   &&+\sum_{z\in\Z_1}P_Z(z)\left[(1-H(Y|Z=z))H(\tilde{P}\otimes h_b^{-1}(0))+H(Y|Z=z)H(\tilde{P}\otimes h_b^{-1}(1))\right]\\
   &\stackrel{(c)}{=}&  \sum_{z\in\Z_0}P_Z(z)\left[(1-H(Y|Z=z))H(X|Y)+H(Y|Z=z)H(X_{\text{unif}})\right]\\
   &&+\sum_{z\in\Z_1}P_Z(z)\left[(1-H(Y|Z=z))H(X|Y)+H(Y|Z=z)H(X_{\text{unif}})\right]\\
   &=& H(X|Y)[1-H(Y|Z)]+H(Y|Z)H(X_{\text{unif}}),
\end{eqnarray*}
where $H(X_{\text{unif}})$ denotes the entropy of $X$ when $Y$ is uniformly distributed. Here, $(a)$ is due to \eqref{Proof_BISO_2} and \eqref{Proof_BISO_3}, $(b)$ follows form convexity of $u\mapsto H(P\otimes h_b^{-1}(u)))$ for all $u\in[0,1]$ \cite{Shamai_BISO3} and Jensen's inequality. In $(c)$, we used the symmetry of channel $P_{X|Y}$ to show  that $H(X|Y=0)=H(X|Y=1)=H(X|Y)$. Hence, we obtain
\begin{equation*}
  H(Y|Z)\geq\frac{H(X|Z)-H(X|Y)}{H(X_{\text{unif}})-H(X|Y)}=\frac{I(X;Y)-I(X;Z)}{C(P_{X|Y})},
\end{equation*}
   where the equality follows from the fact that for BISO channel (and in general for any quasi-symmetric channel) the uniform input distribution is the capacity-achieving distribution \cite[Lemma 4.18]{Fady_Lecture_note}. Since $g_{\eps}(X;Y)$ is attained when $I(X;Z)=\eps$, the conclusion immediately follows.
\end{proof}
%When one makes a further assumption that $X\sim\sBer(\frac{1}{2})$, then Conjecture~\ref{Conjecture_Theorem_BSC_BEC} can be shown to hold. To this end, consider a random variable $Z$ such that $X\markov Y\markov Z$. We can write
%$$H(X|z)=h(\alpha*h^{-1}(H(Y|z)))\leq (1-H(Y|z))h(\alpha)+H(Y|z),$$
% where the inequality follows from convexity of $x\mapsto h(\alpha*h^{-1}(x))$, see for example \cite{mrsgerber}, and Jensen's inequality. It then follows that
% $$H(Y|Z)\geq \frac{H(X|Z)-h(\alpha)}{1-h(\alpha)}$$
%and since $H(X)=H(Y)=1$, we can conclude that
%$$I(Y;Z)\leq \frac{1-H(X|Z)}{1-h(\alpha)}.$$ This together with Corollary~\ref{corollary_lowerbound}, concludes that for any $\eps\leq I(X;Y)$,
%$$g_{\eps}(X;Y)=\frac{\eps}{1-h(\alpha)},$$
%and in general for any $\eps\geq 0$,
%$$g_{\eps}(X;Y)= \min\left\{\frac{\eps}{1-h(\alpha)}, 1\right\},$$ which shows that Conjecture~\ref{Conjecture_Theorem_BSC_BEC} holds in this case.
This lemma then shows that the larger the gap between $I(X;Y)$ and $I(X;Y')$ is for $Y'\sim\sBer(0.5)$, the more $g_{\eps}(X;Y)$ deviates from its lower bound. When $Y\sim \sBer(0.5)$, then $C(P_{Y|X})=I(X;Y)$ and $H(Y)=1$ and hence Lemma~\ref{lemma_BISO} implies that
$$\frac{\eps}{I(X;Y)}\leq g_{\eps}(X;Y)\leq 1-\frac{I(X;Y)-\eps}{I(X;Y)}=\frac{\eps}{I(X;Y)},$$ and hence we have proved the following corollary.
\begin{corollary}\label{corollary_BISO_Uniform}
If the channel $P_{X|Y}$ is BISO and $Y\sim \sBer(0.5)$, then for any $\eps\geq 0$
$$g_{\eps}(X;Y)=\frac{1}{I(X;Y)}\min\{\eps, I(X;Y)\}.$$
\end{corollary}
This corollary now enables us to prove the reverse direction of Theorem~\ref{Theorem_slope-equality} for the family of BISO channels.
\begin{theorem}\label{Theorem_Reverse_Direction_BISO}
If $P_{X|Y}$ is a BISO channel, then the following statements are equivalent:
\begin{itemize}
\item[(i)] $g_{\eps}(X;Y)=\eps\frac{H(Y)}{I(X;Y)}$ for $0\leq\eps\leq I(X;Y)$.
\item[(ii)] The initial efficiency of privacy-constrained information extraction is
$$g'_0(X;Y)=\frac{-\log P_Y(y)}{D(P_{X|Y}(\cdot|y)||P_X(\cdot))}, ~~\forall y\in\Y.$$
\end{itemize}
\end{theorem}
\begin{proof}
\emph{(i)$\Rightarrow$ (ii)}.

 This follows from Theorem~\ref{Theorem_slope-equality}.

\emph{(ii)$\Rightarrow$ (i)}.

Let $Y\sim\sBer(p)$ for $0<p<1$, and, as before, $\X=\{\pm 1, \pm 2, \dots, \pm k\}$, so that $P_{X|Y}$ is determined by a $2\times (2k)$ matrix. We then have
\begin{equation}\label{ratio1}
  \frac{-\log P_Y(0)}{D(P_{X|Y}(\cdot|0)||P_X(\cdot))}=\frac{\log(1-p)}{H(X|Y)+\sum_{x=-k}^kP_{X|Y}(x|0)\log(P_X(x))},
\end{equation}
and
\begin{equation}\label{ratio2}
  \frac{-\log P_Y(1)}{D(P_{X|Y}(\cdot|1)||P_X(\cdot))}=\frac{\log(p)}{H(X|Y)+\sum_{x=-k}^kP_{X|Y}(x|1)\log(P_X(x))}.
\end{equation}
The hypothesis implies that \eqref{ratio1} is equal to \eqref{ratio2}, that is,
\begin{equation}\label{ratio3}
  \frac{\log(1-p)}{H(X|Y)+\sum_{x=-k}^kP_{X|Y}(x|0)\log(P_X(x))}=\frac{\log(p)}{H(X|Y)+\sum_{x=-k}^kP_{X|Y}(x|1)\log(P_X(x))}.
\end{equation}
%Note that
%\begin{eqnarray*}
%% \nonumber to remove numbering (before each equation)
%  H(X|Y)+\sum_{x=-k}^kP_{X|Y}(x|0)\log(P_X(x))&=& H(X|Y)+\sum_{x=-k}^{-1}P_{X|Y}(x|0)\log(P_X(x))+\sum_{x=1}^{k}P_{X|Y}(x|0)\log(P_X(x))\\
%   &=& H(X|Y)+\sum_{x=1}^{k}P_{X|Y}(x|1)\log(P_X(x))+\sum_{x=1}^{k}P_{X|Y}(x|0)\log(P_X(x)) \\
%%   &=& 
%\end{eqnarray*}

It is shown in Appendix~\ref{Appendix_BISO_Uniform} that \eqref{ratio3} holds if and only if $p=0.5$. Now we can invoke Corollary~\ref{corollary_BISO_Uniform} to conclude that $g_{\eps}(X;Y)=\eps\frac{H(Y)}{I(X;Y)}$.
\end{proof}

This theorem shows that for any BISO $P_{X|Y}$ channel with uniform input, the optimal privacy filter is an erasure channel depicted in Fig.\ \ref{fig:LowerBound_Filter}. Note that if $P_{X|Y}$ is a BSC with uniform input $P_{Y}=\sBer(0.5)$, then $P_{Y|X}$ is also a BSC with uniform input $P_{X}=\sBer(0.5)$. The following corollary specializes Corollary~\ref{corollary_BISO_Uniform} for this case.
%
%
%Since any BISO channel is also row-permutation channel \textcolor{red}{CITATION.}, it is easy to show that for BISO $P_{X|Y}$, $D(P_{X|Y}(\cdot|y)||P_X(\cdot))$ does not depend on $y$ and hence by Corollary~\ref{Corollary_Y_Uniform}, if $Y$ is not uniform then $\eps\mapsto g_{\eps}(X;Y)$ can not be linear. This is because, as we will see in the sequel, in this case $g_0(X;Y)=0$, and hence Corollary~\ref{Corollary_Y_Uniform} implies that for linear $g_{\eps}(X;Y)$, the fact that $D(P_{X|Y}(\cdot|y)||P_X(\cdot))$ is constant for any $y\in\Y$ is equivalent of $Y$ being uniform. This statement is clear best in the case of binary symmetric channel. Assume that private source $X$ is binary and observation data is $Y=X\oplus M$ where $M\sim \sBer(\delta)$. In other words, the channel from $X$ to $Y$ is a binary symmetric channel with crossover probability $0\leq\alpha< 0.5$; BSC$(\alpha)$.

\begin{corollary}\label{corollary_BSC_BEC}
For the joint distribution $P_{X}P_{Y|X}=\sBer(0.5)\times \text{BSC}(\alpha)$, the binary erasure channel with erasure probability (shown in Fig.~\ref{fig:Bsc})
\begin{equation}\label{BEC_optimal_prob}
    \delta(\eps, \alpha):=1-\frac{\eps}{I(X;Y)}.
\end{equation}
for $0\leq \eps\leq I(X;Y)$, is the optimal privacy filter in \eqref{gepsilon}. In other words, for $\eps\geq 0$
 $$g_{\eps}(X;Y)=\frac{1}{I(X;Y)}\min\{\eps, I(X;Y)\}.$$ Moreover, for a given $0<\alpha<\frac{1}{2}$, $P_X=\sBer(0.5)$ is the only distribution for which $\eps\mapsto g_{\eps}(X;Y)$ is linear. That is, for $P_{X}P_{Y|X}=\sBer(p)\times \text{BSC}(\alpha)$, $0<p<0.5$, we have
 $$g_{\eps}(X;Y)>\eps\frac{H(Y)}{I(X;Y)}.$$
\end{corollary}
\begin{proof}\label{Remark_on_Uniform_binary}
As mentioned earlier, since $P_X=\sBer(0.5)$ and $P_{Y|X}$ is $\text{BSC}(\alpha)$, it follows that $P_{X|Y}$ is also a BSC with uniform input and hence from Corollary~\ref{corollary_BISO_Uniform}, we have $g_{\eps}(X;Y)=\frac{\eps}{I(X;Y)}$. As in this case $g_{\eps}(X;Y)$ achieves the lower bound given in Lemma~\ref{non-increasing-Lemma}, we conclude from Fig.~\ref{fig:LowerBound_Filter} that BEC($\delta(\eps, \alpha)$), where $\delta(\eps, \alpha)=1-\frac{\eps}{I(X;Y)}$, is an optimal privacy filter.  The fact that $P_X=\sBer(0.5)$ is the only input distribution for which $\eps\mapsto g_{\eps}(X;Y)$ is linear follows from the proof of Theorem~\ref{Theorem_Reverse_Direction_BISO}. In particular, we saw that a necessary and sufficient condition for $g_{\eps}(X;Y)$ being linear is that the ratio $\frac{-\log P_Y(y)}{D(P_{X|Y}(\cdot|y)||P_X(\cdot))}$ is constant for all $y\in\Y$. As shown before, this is equivalent to $Y\sim\sBer(0.5)$. For the binary symmetric channel, this is equivalent to $X\sim\sBer(0.5)$.
\end{proof}
 \vspace{-5 mm}
\begin{figure}[h]
\centering
\begin{tikzpicture}
%\tikzstyle{every node} = [circle, fill = gray!30]
\node (a) [circle] at (0,0) {1};
\node (b) [circle] at (0,1.4) {0};
\node (c) [circle] at (2,0) {1};
\node (d) [circle] at (2,1.4) {0};
\draw[->] (0.15,0) -- (1.85,0) node[pos=.5,sloped,below] {\footnotesize{$1-\alpha$}};
\draw[->] (0.15,0) -- (1.85,1.4) node[pos=.75,sloped,below] {};
\draw[->] (0.15,1.4) -- (1.85,0) node[pos=.25,sloped,above] {};
\draw[->] (0.15,1.4) -- (1.85,1.4) node[pos=.5,sloped,above] {\footnotesize{$1-\alpha$}};
\node (e) [circle] at (2.3,0) {};
\node (f) [circle] at (2.3,1.4) {};
\node (g) [circle] at (4.6,0) {1};
\node (h) [circle] at (4.6,1.4) {0};
\node (k) [circle] at  (4.6, 0.7) {\text{e}};
\draw[->] (2.3,0) -- (4.45,0) node[pos=.5,sloped,below] {\footnotesize{$1-\delta(\eps,\alpha)$}};
\draw[->] (2.3,0) -- (4.45, 0.7) node[pos=.5,sloped,below] {};
\draw[->] (2.3,1.4) -- (4.45,1.4) node[pos=.5,sloped,above] {\footnotesize{$1-\delta(\eps,\alpha)$}};
\draw[->] (2.3,1.4) -- (4.45, 0.7) node[pos=.5,sloped,above] {};
\end{tikzpicture}
\caption{\small{Optimal privacy filter for $P_{Y|X}=BSC(\alpha)$ with uniform $X$ where $\delta(\eps,\alpha)$ is specified in \eqref{BEC_optimal_prob}}.} \label{fig:Bsc}
\end{figure}
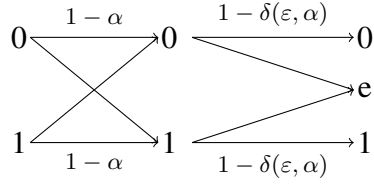
\vspace{1 mm}
  The optimal privacy filter for BSC($\alpha$) and uniform $X$ is shown in Figure~\ref{fig:Bsc}.   In fact, this corollary immediately implies that the general lower-bound given in \eqref{UB&LB} is tight for the binary symmetric channel with uniform $X$.

\subsubsection{Erasure Observation Channel}

%In general, we know that for any Markov chain $X\markov Y\markov Z$ we have
%$$I(Y;Z)=I(X,Y;Z)=I(X;Z)+I(Y;Z|X),$$ and thus
%\begin{equation}\label{Loose_UB}
%    g_{\eps}(X;Y)=H(Y|X)+\eps -\inf_{P_{Z|Y}\in \D_{\eps}} H(Y|Z,X)\leq      H(Y|X)+\eps.
%\end{equation}
Combining \eqref{General_Lower_Bound} and Lemma~\ref{non-increasing-Lemma}, we have for $\eps\leq I(X;Y)$
\begin{equation}\label{General_LB_UB}
  \eps\frac{H(Y)}{I(X;Y)}+g_0(X;Y)\left[1-\frac{\eps}{I(X;Y)}\right]\leq g_{\eps}(X;Y)\leq H(Y|X)+\eps,
\end{equation}

In the following we show that the above upper and lower bound coincide when $P_{Y|X}$ is an erasure channel, i.e.,  $P_{Y|X}(x|x)=1-\delta$ and $P_{Y|X}(\text{e}|x)=\delta$ for all $x\in \X$ and  $0\leq\delta\leq 1$.

\begin{lemma}\label{Lemma_erasure}
For any given $(X,Y)$, if $P_{Y|X}$ is an erasure channel (as defined above), then $$g_{\eps}(X;Y)=H(Y|X)+\min\{\eps, I(X;Y)\},$$ for any $\eps\geq 0$.
\end{lemma}
\begin{proof}
It suffices to show that if $P_{Y|X}$ is an erasure channel, then $g_0(X;Y)=H(Y|X)$. This follows, since if $g_0(X;Y)=H(Y|X)$, then the lower bound in \eqref{General_LB_UB} becomes $H(Y|X)+\eps$ and thus $g_{\eps}(X;Y)=H(Y|X)+\eps$.

Let $|\X|=m$ and $\Y=\X\cup \{\text{e}\}$ where $\text{e}$ denotes the erasure symbol. Consider the following privacy filter to generate $Z\in \Y$:
$$P_{Z|Y}(z|y)=\left\{
  \begin{array}{ll}
   \frac{1}{m}  & \text{if} ~~y\neq \text{e}, z\neq \text{e}, \\
    1 & \text{if} ~~ y=z=\text{e}.
  \end{array}
\right.$$

For any $x\in \X$, we have
$$P_{Z|X}(z|x)=P_{Z|Y}(z|x)P_{Y|X}(x|x)+P_{Z|Y}(z|\text{e})P_{Y|X}(\text{e}|x)=\left[\frac{1-\delta}{m}\right]1_{\{z\neq \text{e}\}}+\delta1_{\{z=\text{e}\}},$$
which implies $Z\indep X$ and thus $I(X;Z)=0$. On the other hand, $P_Z(z)=\left[\frac{1-\delta}{m}\right]1_{\{z\neq \text{e}\}}+\delta 1_{\{z=\text{e}\}}$, and therefore we have
\begin{eqnarray*}
  g_0(X;Y)&\geq & I(Y;Z) = H(Z)-H(Z|Y)=H\left(\frac{1-\delta}{m}, \dots, \frac{1-\delta}{m}, \delta\right)-(1-\delta)\log (m) \\
   &=& h(\delta)=H(Y|X).
\end{eqnarray*}
 It then follows from Lemma~\ref{non-increasing-Lemma} that $g_0(X;Y)=H(Y|X)$, which completes the proof.
%Consider the channel $P_{Z|Y}$ which connects the erasure symbol $Y=\text{e}$ to $Z=\text{e}$ with probability one and connects the rest $Y=y$ to $Z=z$ with probabilities such that the privacy constraint $I(X;Z)\leq \eps$ is satisfied. First note that $P_{X|Z}(x|\text{e})=P_X(x)$ and hence connecting $Y=e$ to $Z=e$ does not violate the privacy constraint. The fact that we can always find channel $P_{Z|Y}(z|y)$ for $z\neq \text{e}$ and $y\neq \text{e}$ such that $I(X;Z)\leq \eps$, stems from the continuity of mutual information.
%Such a channel for binary $X$ and $Y$ is depicted in Fig.~\ref{fig:Bec}. The result then follows from \eqref{Loose_UB} because in this case $H(Y|X,Z)=0$ which makes the inequality in \eqref{Loose_UB} tight.
\end{proof}
\begin{example}
 In light of this lemma, we can conclude that if $P_{Y|X}=\text{BEC}(\delta)$, then the optimal privacy filter is a combination of an identity channel and  a BSC($\alpha(\eps, \delta)$), as shown in Fig. \ref{fig:Bec}, where $0\leq\alpha(\eps,\delta)\leq~\frac{1}{2}$ is the unique solution of
\begin{equation}\label{BSC_optimal_prob}
    (1-\delta)[h_b(\alpha*p)-h_b(\alpha)]=\eps,
\end{equation}
 where $X\sim\sBer(p)$, $p\leq 0.5$ and $a*b=a(1-b)+b(1-a)$. Note that it is easy to check that $I(X;Z)=(1-\delta)[h_b(\alpha*p)-h_b(\alpha)]$. Therefore, in order for this channel to be a valid privacy filter, the crossover probability, $\alpha(\eps, \delta)$, must be chosen such that $I(X;Z)=\eps$. We note that for fixed $0< \delta< 1$ and $0< p< 0.5$, the map $\alpha\mapsto (1-\delta)[h_b(\alpha*p)-h_b(\alpha)]$ is monotonically decreasing on $[0, \frac{1}{2}]$ ranging over $[0, (1-\delta)h_b(p)]$ and since $\eps\leq I(X;Y)=(1-\delta)h_b(p)$, the solution of the above equation is unique.
\end{example}
\vspace{-2 mm}
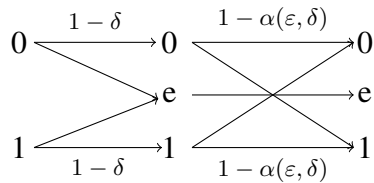
\begin{figure}[h]
\centering
\begin{tikzpicture}
%\tikzstyle{every node} = [circle, fill = gray!30]
\node (a) [circle] at (0,0) {1};
\node (b) [circle] at (0,1.4) {0};
\node (c) [circle] at (2,0) {1};
\node (d) [circle] at (2,1.4) {0};
\node (e) [circle] at (2,0.7) {e};
\draw[->] (0.2,0) -- (1.9,0) node[pos=.5,sloped,below] {\footnotesize{$1-\delta$}};
\draw[->] (0.2,0) -- (1.85,0.65) node[pos=.75,sloped,below] {};
\draw[->] (0.2,1.4) -- (1.85,0.65) node[pos=.25,sloped,above] {};
\draw[->] (0.2,1.4) -- (1.82,1.4) node[pos=.5,sloped,above] {\footnotesize{$1-\delta$}};
\node (e1) [circle] at (2.3,0) {};
\node (e2) [circle] at (2.3,0.7) {};
\node (f) [circle] at (2.3,1.4) {};
\node (g) [circle] at (4.6,0) {1};
\node (h) [circle] at (4.6,1.4) {0};
\node (k) [circle] at  (4.6, 0.7) {\text{e}};
\draw[->] (2.3,0) -- (4.45,0) node[pos=.5,sloped,below] {\footnotesize{$1-\alpha(\eps,\delta)$}};
\draw[->] (2.3,0) -- (4.45,1.4) node[pos=.5,sloped,below] {};
\draw[->] (2.3,1.4)-- (4.45,1.4) node[pos=.5,sloped,above]{\footnotesize{$1-\alpha(\eps,\delta)$}};
\draw[->] (2.3,0.7) -- (4.45, 0.7) node[pos=.5,sloped,above] {};
\draw[->] (2.3,1.4) -- (4.45,0) node[pos=.5,sloped,above] {};
\end{tikzpicture}
\caption{\small{Optimal privacy filter for $P_{Y|X}=BEC(\delta)$ where $\delta(\eps,\alpha)$ is specified in \eqref{BSC_optimal_prob}}.} \label{fig:Bec}
\end{figure}
Combining Lemmas~\ref{non-increasing-Lemma} and \ref{Lemma_erasure} with Corollary~\ref{corollary_BSC_BEC}, we can show the following \emph{extremal property} of the BEC and BSC channels, which is similar to other existing extremal properties of the BEC and the BSC, see e.g., \cite{Shamai_BISO2} and \cite{Shamai_BISO1}. For $X\sim\sBer(0.5)$, we have for any channel $P_{Y|X}$,
$$g_{\eps}(X;Y)\geq \frac{\eps}{I(X;Y)}=g_{\eps}(\text{BSC}(\hat{\alpha})),$$ where $g_{\eps}(\text{BSC}(\alpha))$ is the rate-privacy function corresponding to $P_{XY}=\sBer(0.5)\times \text{BSC}(\alpha)$ and  $\hat{\alpha}:=h_b^{-1}(H(X|Y))$. Similarly, if $X\sim \sBer(p)$, we have for any channel $P_{Y|X}$ with $H(Y|X)\leq 1$,
$$g_{\eps}(X;Y)\leq H(Y|X)+\eps=g_{\eps}(\text{BEC}(\hat{\delta})),$$
where $g_{\eps}(\text{BEC}(\delta))$ is the rate-privacy function corresponding to $P_{XY}=\sBer(p)\times \text{BEC}(\delta)$ and  $\hat{\delta}:=h_b^{-1}(H(Y|X))$.
\section{Rate-Privacy Function for Continuous Random Variables}

In this section we extend the rate-privacy function $g_{\eps}(X;Y)$ to the continuous case. Specifically, we assume that the private and observable data are continuous random variables and that the filter is composed of two stages: first Gaussian noise is added and then the resulting random variable is quantized using an $M$-bit accuracy uniform scalar quantizer (for some positive integer $M\in\N$). These filters are of practical interest as they can be easily implemented. This section is divided in two subsections, in the first we discuss general properties of the rate-privacy function and in the second we study the Gaussian case in more detail. Some observations on $\hat{g}_{\eps}(X;Y)$ for continuous $X$ and $Y$ are also given.

\subsection{General properties of the rate-privacy function}
\label{Subsection:gContinuous}

Throughout this section we assume that the random vector $(X,Y)$ is absolutely continuous with respect to the Lebesgue measure on $\R^2$. Additionally, we assume that its joint density $f_{X,Y}$ satisfies the following.
\begin{itemize}
	\item[(a)] There exist constants $C_1>0$, $p>1$ and bounded function $C_2:\R\to \R$ such that
$$f_Y(y)\leq C_1|y|^{-p},$$ and also for $x\in \R$ $$f_{Y|X}(y|x)\leq C_2(x)|y|^{-p},$$

%\item[(a)] the map $(x,y)\mapsto f_{X,Y}(x,y)$, $x\mapsto f_X(x)$, and $y\mapsto f_Y(y)$ are bounded. We also assume that there exist constants $C>0$ and $p>1$ such that
%$$f_Y(y)\leq C|y|^{-p},$$
	\item[(b)] $\E[X^2]$ and $\E[Y^2]$ are both finite,
	\item[(c)] the differential entropy of $(X,Y)$ satisfies $h(X,Y)>-\infty$,
\item[(d)] $H(\lfloor Y\rfloor)<\infty$, where $\lfloor a\rfloor$ denotes the  largest integer $\ell$ such that $\ell\leq a$.
%	\begin{equation*}
%	x\mapsto f_X(x) |\log f_X(x)|, \;y\mapsto f_Y(y) |\log f_Y(y)|\textnormal{ and }(x,y)\mapsto f_{X,Y}(x,y)|\log(f_{X,Y}(x,y))|
%	\end{equation*}
\end{itemize}

Note that assumptions (b) and (c) together imply that $h(X, Y)$, $h(X)$ and $h(Y)$ are finite, i.e., the maps $x\mapsto f_X(x) |\log f_X(x)|, \;y\mapsto f_Y(y) |\log f_Y(y)|$ and $(x,y)\mapsto f_{X,Y}(x,y)|\log(f_{X,Y}(x,y))|$ are integrable.
We also assume that $X$ and $Y$ are not independent, since otherwise the problem to characterize $g_{\eps}(X;Y)$ becomes trivial by assuming that the displayed data $Z$ can equal the observable data $Y$.

We are interested in filters of the form $\Q_M(Y+\gamma N)$ where $\gamma\geq0$, $N\sim N(0,1)$ is a standard normal random variable which is independent of $X$ and $Y$, and for any positive integer $M$, $\Q_M$ denotes the $M$-bit accuracy uniform scalar quantizer, i.e., for all $x\in\R$
\begin{equation*}
\Q_M(x) = \frac{1}{2^M} \left\lfloor 2^M x\right\rfloor.
\end{equation*}

Let $Z_\gamma=Y+\gamma N$ and $Z_\gamma^M=\Q_M(Z_\gamma)=\Q_M(Y+\gamma N)$. We define, for any $M\in\N$,
\begin{equation}
\label{eq:ContinuousCaseDefgM}
g_{\eps, M}(X;Y) := \sup\limits_{\begin{smallmatrix}\gamma\geq0,\\I(X;Z_\gamma^M)\leq\eps\end{smallmatrix}} I(Y;Z_\gamma^M),
\end{equation}
and similarly
\begin{equation}
\label{eq:ContinuousCaseDefg}
g_\eps(X;Y) :=  \sup\limits_{\begin{smallmatrix}\gamma\geq0,\\I(X;Z_\gamma)\leq\eps\end{smallmatrix}} I(Y;Z_\gamma).
\end{equation}

The next theorem shows that the previous definitions are closely related.

\begin{theorem}
\label{Thm:gMtog}
Let $\eps>0$ be fixed. Then $\displaystyle \lim_{M\to\infty} g_{\eps, M}(X;Y) = g_\eps(X;Y)$.
\end{theorem}

\begin{proof}
See Appendix \ref{Appendix:ProofsContinuous}.
\end{proof}

In the limit of large $M$, $g_\eps(X;Y)$ approximates $g_{\eps, M}(X;Y)$. This becomes relevant when $g_\eps(X;Y)$ is easier to compute than $g_{\eps, M}(X;Y)$, as demonstrated in the following subsection. The following theorem summarizes some general properties of $g_\eps(X;Y)$.

\begin{theorem}
\label{Thm:gContinuous}
The function $\eps\mapsto g_\eps(X;Y)$ is non-negative, strictly-increasing, and satisfies
\begin{equation*}
\lim_{\eps\to0} g_\eps(X;Y)=0\quad\quad\quad\textnormal{and}\quad\quad\quad g_{I(X;Y)}(X;Y)=\infty.
\end{equation*}
%Moreover, for every $\eps>0$ there exists a unique $\gamma=\gamma_\eps\geq0$ such that $I(X;Z_\gamma)=\eps$ and $I(Y;Z_\gamma)=g_\eps(X;Y)$.
\end{theorem}

\begin{proof}
See Apendix \ref{Appendix:ProofsContinuous}.
\end{proof}

As opposed to the discrete case, in the continuous case $g_\eps(X;Y)$ is no longer bounded. In the following section we show that $\eps\mapsto g_\eps(X;Y)$ can be convex, in contrast to the discrete case where it is always concave.

We can also define $\hat{g}_{\eps, M}(X;Y)$ and $\hat{g}_\eps(X;Y)$ for continuous $X$ and $Y$, similar to \eqref{eq:ContinuousCaseDefgM} and \eqref{eq:ContinuousCaseDefg}, but where the privacy constraints are replaced by $\rho_m^2(X;Z_{\gamma}^M)\leq \eps$ and $\rho_m^2(X;Z_{\gamma})\leq \eps$, respectively. It is clear to see from Theorem~\ref{Thm:gContinuous} that $\hat{g}_{0}(X;Y)=g_{0}(X;Y)=0$ and $\hat{g}_{\rho^2(X;Y)}(X;Y)=\infty$. However, although we showed that $g_{\eps}(X;Y)$ is indeed the asymptotic approximation of $g_{\eps, M}(X;Y)$ for $M$ large enough, it is not clear that the same statement holds for $\hat{g}_{\eps}(X;Y)$ and $\hat{g}_{\eps, M}(X;Y)$.
%\begin{remark}
%Similar to a information preservation problem introduced in \cite{Kairouz_PHD}, where the maximization of the mutual information between the input and the output of a locally differentially private mechanism \cite{privacyaware} is studied, we can consider the same problem where the mechanism satisfies the MMSE privacy described in \eqref{Privacy_constraint_MMSE}
%$$\hat{g}^f_\eps(X;X):=  \sup\limits_{\begin{smallmatrix}\gamma\geq0,\\\rho_m^2(X;V^f_\gamma)\leq\eps\end{smallmatrix}} I(X;V^f_\gamma),$$
%where, $V^f_{\gamma}:=X+\gamma M_f$ and $M_f$ is a noise process with a stable distribution with density $f$ and parameter $\alpha\in(0, 2]$ (see, \cite[Pages 153-156]{Durret'sBook}).
%\end{remark}

%\textcolor{red}{Like in the discrete case, we can also measure privacy using maximal correlation and define
%\begin{equation}
%\label{eq:ContinuousCaseDefgM2}
%\hat{g}_{\eps, M}(X;Y) = \sup\limits_{\begin{smallmatrix}\gamma\geq0,\\\rho_m^2(X;Z_\gamma^M)\leq\epsilon\end{smallmatrix}} I(Y;Z_\gamma^M),
%\end{equation}
%and similarly
%\begin{equation}
%\label{eq:ContinuousCaseDefg2}
%\hat{g}_\eps(X;Y) =  \sup\limits_{\begin{smallmatrix}\gamma\geq0,\\\rho_m^2(X;Z_\gamma)\leq\epsilon\end{smallmatrix}} I(Y;Z_\gamma).
%\end{equation}
%TALK ABOUT PROPERTIES OF THESE TWO FUNCTIONS
%}

\subsection{Gaussian Information}

The rate-privacy function for Gaussian $Y$ has an interesting interpretation from an estimation theoretic point of view. Given the private and observable data $(X,Y)$, suppose an agent is required to \emph{estimate} $Y$ based on the output of the privacy filter. We wish to know the effect of imposing a privacy constraint on the estimation performance.

The following lemma shows that $g_{\eps}(X;Y)$ bounds the best performance of the predictability of $Y$ given the output of the privacy filter. The proof provided for this lemma does not use the Gaussianity of the noise process, so it holds for any noise process.
\begin{lemma}\label{Lemma_MMSE_Gaussian}
For any given private data $X$ and Gaussian observable data $Y$, we have for any $\eps\geq 0$
$$\inf_{\begin{smallmatrix}\gamma\geq0,\\I(X;Z_\gamma)\leq\eps\end{smallmatrix}}\mmse(Y|Z_\gamma)\geq \var(Y)2^{-2g_{\eps}(X;Y)}.$$
\end{lemma}
\begin{proof}
%Since $Y$ and $Z_{\gamma}=Y+\gamma N$ are both Gaussian, it is straightforward to see that in this case \begin{equation}\label{mutualinfo}
%    \mmse(Y|Z_{\gamma}) = \var(Y)2^{-2I(Y;Z_{\gamma})}.
%\end{equation}
%Taking the infimum from both sides of the above equality over all $\gamma$ such that $I(X; Z_{\gamma})\leq \eps$, we can obtain the result.
It is a well-known fact from rate-distortion theory that  for a Gaussian $Y$ and its reconstruction $\hat{Y}$
$$I(Y; \hat{Y})\geq \frac{1}{2}\log\frac{\mathsf{var}(Y)}{\mathbb{E}[(Y-\hat{Y})^2]},$$ and hence by setting $\hat{Y}=\mathbb{E}[Y|Z_{\gamma}]$, where $Z_{\gamma}$ is an output of a privacy filter, and noting that $I(Y; \hat{Y})\leq I(Y; Z_{\gamma})$, we obtain
\begin{equation}\label{mutualinfo}
    \mmse(Y|Z_{\gamma}) \geq \var(Y)2^{-2I(Y;Z_{\gamma})},
\end{equation}
from which the result follows immediately.
\end{proof}

According to Lemma~\ref{Lemma_MMSE_Gaussian}, the quantity $\lambda_{\eps}(X):=2^{-2g_{\eps}(X;Y)}$ is a parameter that bounds the difficulty of estimating Gaussian $Y$ when observing an additive perturbation $Z$ with privacy constraint $I(X;Z)\leq \eps$. Note that $0< \lambda_{\eps}(X)\leq 1$, and therefore, provided that the privacy threshold is not trivial (i.e, $\eps<I(X;Y)$), the mean squared error of estimating $Y$ given the privacy filter output is bounded away from zero, however the bound decays exponentially at rate of $g_{\eps}(X;Y)$.

To finish this section, assume that $X$ and $Y$ are jointly Gaussian with correlation coefficient $\rho$. The value of $g_{\eps}(X;Y)$ can be easily obtained in closed form as demonstrated in the following theorem.
%\begin{theorem}[\cite{Oohama}]\label{Theorem_oohama}
% Let $(X,Y)$ be a pair of jointly Gaussian random variables with zero mean and correlation coefficient $\rho$. Then, for any $Z$ that satisfies Markov relation $X\markov Y\markov Z$, we have:
%$$2^{-2I(X;Z)}\geq 1-\rho^2+\rho^2 2^{-2I(Y;Z)}, $$ where equality occurs if $(X, Y, Z)$ are jointly Gaussian.
%\end{theorem}

\begin{theorem}\label{Thm:G_gaussian}
Let $(X,Y)$ be jointly Gaussian random variables with correlation coefficient $\rho$. For any $\eps\in[0,I(X;Y))$ we have $$g_{\eps}(X;Y) = \frac{1}{2}\log\left(\frac{\rho^2}{2^{-2\eps}+\rho^2-1}\right).$$
\end{theorem}
\begin{proof}
 One can always write $Y=aX+N_1$ where $a^2=\rho^2\frac{\var(Y)}{\var(X)}$ and $N_1$ is a Gaussian random variable with mean $0$ and variance $\sigma^2=(1-\rho^2)\var(Y)$ which is independent
of $(X,Y)$. On the other hand, we have $Z_{\gamma}=Y+\gamma N$ where $N$ is the standard Gaussian random variable independent of $(X,Y)$ and hence $Z_{\gamma}=aX+N_1+\gamma N$. In order for this additive channel to be a privacy filter, it must satisfy
$$I(X;Z_{\gamma})\leq \eps,$$ which implies $$\frac{1}{2}\log\left(\frac{\var(Y)+\gamma^2}{\sigma^2+\gamma^2}\right)\leq \eps,$$ and hence
$$\gamma^2\geq \frac{2^{-2\eps}+\rho^2-1}{1-2^{-2\eps}}\var(Y)=:\gamma^*.$$
Since $\gamma\mapsto I(Y; Z_{\gamma})$ is strictly decreasing (cf., Appendix~\ref{Appendix:ProofsContinuous}), we obtain
\begin{eqnarray}
  g_{\eps}(X;Y)&=&I(Y; Z_{\gamma^*})= \frac{1}{2}\log\left(1+\frac{\var(Y)}{\gamma^2}\right)\nonumber\\
   &=&  \frac{1}{2}\log\left(1+\frac{1-2^{-2\eps}}{2^{-2\eps}+\rho^2-1}\right)\label{Proof_Gaussian_gEpsilon}.\qedhere
\end{eqnarray}
\end{proof}
According to \eqref{Proof_Gaussian_gEpsilon}, we conclude that the optimal privacy filter for jointly Gaussian $(X,Y)$ is an additive Gaussian channel with signal to noise ratio $\dfrac{1-2^{-2\eps}}{2^{-2\eps}+\rho^2-1}$, which shows that if perfect privacy is required, then the displayed data is independent of the observable data $Y$, i.e., $g_0(X;Y)=0$.
%\begin{equation*}
%g_\eps(X;Y) = \frac{1}{2} \log\left(1+\frac{1-2^{-2\eps}}{2^{-2\eps}+\rho^2-1}\right).
%\end{equation*}

%\begin{figure}[hbtp]
%\caption{Rate-privacy function for several values of $\rho$.}
%\centering
%\includegraphics[scale=0.3]{Privacy.png}
%\end{figure}
\begin{remark}
We could assume that the privacy filter adds non-Gaussian noise to the observable data and define the rate-privacy function accordingly. To this end, we define
$$g_{\eps}^f(X;Y):=\sup_{\gamma\geq 0, \atop I(X;Z^f_{\gamma})}I(Y; Z_{\gamma}^f),$$ where $Z_{\gamma}^f=Y+\gamma M_f$ and $M_f$ is a noise process that has stable distribution with density $f$ and is independent of $(X,Y)$. In this case, we can use a technique similar to Oohama \cite{Oohama} to lower bound $g_{\eps}^f(X;Y)$ for jointly Gaussian $(X,Y)$. Since $X$ and $Y$ are jointly Gaussian, we can write $X=aY+bN$ where $a^2=\rho^2\frac{\var(X)}{\var(Y)}$, $b=\sqrt{(1-\rho^2)\var{X}}$, and $N$ is a standard Gaussian random variable that is independent of $Y$.
We can apply the conditional entropy power inequality (cf., \cite[Page 22]{networkinfotheory}) for a random variable $Z$ that is independent of $N$, to obtain
$$2^{2h(X|Z)}\geq 2^{2h(aY|Z)}+2^{2h(N)}=a^22^{2h(Y|Z)}+2\pi e(1-\rho^2)\var(X),$$
and hence
$$2^{-2I(X;Z)}2^{2h(X)}\geq a^22^{2h(Y)}2^{-2I(Y;Z)}+2\pi e(1-\rho^2)\var(X).$$
Assuming $Z=Z^f_{\gamma}$ and taking infimum from both sides of above inequality over $\gamma$ such that $I(X;Z_{\gamma}^f)\leq \eps$,  we obtain
$$g^f_{\eps}(X;Y)\geq \frac{1}{2}\log\left(\frac{\rho^2}{2^{-2\eps}+\rho^2-1}\right)=g_{\eps}(X;Y),$$
which shows that for Gaussian $(X,Y)$, Gaussian noise is the worst stable additive noise in the sense of privacy-constrained information extraction.
\end{remark}

We can also calculate $\hat{g}_{\eps}(X;Y)$ for jointly
Gaussian $(X,Y)$.
\begin{theorem}\label{Thm:Ghat_Gaussian}
Let $(X,Y)$ be jointly Gaussian random variables with correlation coefficient $\rho$. For any $\eps\in[0,\rho^2)$ we have that $$\hat{g}_{\eps}(X;Y) = \frac{1}{2}\log\left(\frac{\rho^2}{\rho^2-\eps}\right).$$
\end{theorem}
\begin{proof}
Since for the correlation coefficient between $Y$ and $Z_{\gamma}$ we have for any $\gamma\geq 0$, $$\rho^2(Y;Z_{\gamma})=\frac{\var(Y)}{\var(Y)+\gamma^2},$$ we can conclude that $$\rho^2(X;Z_{\gamma})=\frac{\rho^2\var(Y)}{\var(Y)+\gamma^2}.$$ Since $\rho_m^2(X;Z)=\rho^2(X;Z)$ (see e.g., \cite{Renyi-dependence-measure}), the privacy constraint $\rho_m^2(X;Z)\leq \eps$ implies that
$$\frac{\rho^2\var(Y)}{\var(Y)+\gamma^2}\leq \eps,$$ and hence
$$\gamma^2\geq \frac{(\rho^2-\eps)\var(Y)}{\eps}=:\hat{\gamma}^2_{\eps}.$$ By monotonicity of the map $\gamma\mapsto I(Y;Z_{\gamma})$, we have
$$\hat{g}_{\eps}(X;Y)=I(Y; Z_{\hat{\gamma}_{\eps}})=\frac{1}{2}\log\left(1+\frac{\var(Y)}{\hat{\gamma}_{\eps}^2}\right)=\frac{1}{2}\log\left(\frac{\rho^2}{{\rho^2-\eps}}\right).\qedhere$$
\end{proof}
Theorems \ref{Thm:G_gaussian} and \ref{Thm:Ghat_Gaussian} show that unlike to the discrete case (cf. Lemmas \ref{Lemma_Concavity} and \ref{Lemma_concavity_gHat}), $\eps\mapsto g_\eps(X;Y)$ and $\eps\mapsto \hat{g}_\eps(X;Y)$ are convex.

%%%%%%%%%%%%%%%%%%%%%%%%%%%%%%%%%%%%%%%%%

\section{Conclusions}

In this paper, we studied the problem of determining the maximal amount of information that one can extract by observing a random variable $Y$, which is correlated with another random variable $X$ that represents sensitive or private data, while ensuring that the extracted data $Z$ meets a privacy constraint with respect to $X$. Specifically, given two correlated discrete random variables $X$ and $Y$, we introduced the rate-privacy function as the maximization of $I(Y;Z)$ over all stochastic ''privacy filters'' $P_{Z|Y}$ such that $pm(X;Z) \leq \epsilon$, where $pm(\cdot;\cdot)$ is a privacy measure and $\epsilon\geq0$ is a given privacy threshold. We considered two possible privacy measure functions, $pm(X;Z)=I(X;Z)$ and $pm(X;Z)=\rho_m^2(X;Z)$ where $\rho_m$ denotes maximal correlation, resulting in the rate-privacy functions $g_{\epsilon}(X;Y)$ and $\hat{g}_{\epsilon}(X;Y)$, respectively.  We analyzed these two functions, noting that each function lies between easily evaluated upper and lower bounds, and derived their monotonicity and concavity properties. We next provided an information-theoretic interpretation for $g_{\epsilon}(X;Y)$ and an estimation-theoretic characterization for $\hat{g}_{\epsilon}(X;Y)$. In particular, we demonstrated that the dual function of $g_{\epsilon}(X;Y)$ is a corner point of an outer bound on the achievable region of the dependence dilution coding problem. We also showed that $\hat{g}_{\epsilon}(X;Y)$ constitutes the largest amount of information that can be extracted from $Y$ such that no meaningful MMSE estimation of any function of $X$ can be realized by just observing the extracted information $Z$. We then examined conditions on $P_{XY}$ under which the lower bound on $g_{\epsilon}(X;Y)$ is tight, hence determining the exact value of $g_{\epsilon}(X;Y)$. We also showed that for any given $Y$, if the observation channel $P_{Y|X}$ is an erasure channel, then $g_{\epsilon}(X;Y)$ attains its upper bound. Finally, we extended the notions of the rate-privacy functions $g_{\epsilon}(X;Y)$ and $\hat{g}_{\epsilon}(X;Y)$ to the continuous case where the observation channel consists of an additive Gaussian noise channel followed by uniform scalar quantization.

%%%%%%%%%%%%%%%%%%%%%%%%%%%%%%%%%%%%%%%%%%
% INCLUDE THESE IN THE FINAL PAPER
%\acknowledgments{Acknowledgments}
%
%
%
%%%%%%%%%%%%%%%%%%%%%%%%%%%%%%%%%%%%%%%%%%%
%
%\authorcontributions{Author Contributions}
%
%Required if more than one author. Authorship must include and be strictly limited to researchers who have substantially contributed to the reported work. Please carefully review our criteria regarding the Qualification for Authorship: \web /instructions.
%
%%%%%%%%%%%%%%%%%%%%%%%%%%%%%%%%%%%%%%%%%%%
%
%\conflictofinterests{Conflicts of Interest}
%
%Required. State any potential conflicts of interest here or ``The authors declare no conflict of interest''.

%=================================================================
% References: Variant A
%=================================================================
% Back Matter (References and Notes)
%----------------------------------------------------------
% Style and layout of the references
%\bibliographystyle{alpha}
%\bibliography{bibliography}
\makeatletter
\renewcommand\@biblabel[1]{#1. }
\makeatother
\bibliography{bibliography}
\bibliographystyle{plain}

\appendices
\section{Proof of Lemma~\ref{Lemma_slope_lowerbound}}\label{AppendixI_Theorem_gHat_lowerbound}

Given a joint distribution $P_{XY}$ defined over $\X\times \Y$ where $\X=\{1,2,\dots, m\}$ and $\Y=\{1,2,\dots, n\}$ with $n\leq m$, we consider a privacy filter specified by the following distribution for $\delta>0$ and $\Z=\{k, e\}$
\begin{eqnarray}
  P_{Z|Y}(k|y) &=& \delta 1_{\{y=k\}} \label{filter1}\label{filter_derivative}\\
  P_{Z|Y}(\text{e}|y) &=& 1-\delta 1_{\{y=k\}}\label{filter_derivative2}
\end{eqnarray}
where $1_{\{\cdot\}}$ denotes the indicator function. The system of $X\markov Y\markov Z$ in this case is depicted in Fig.~\ref{fig:Derivate_gEpsilon} for the case of $k=1$.
\def\bottom#1#2{\hbox{\vbox to #1{\vfill\hbox{#2}}}}
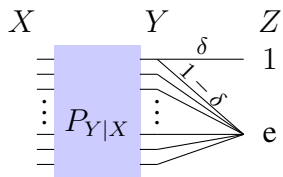
\begin{figure}[!h]
\centering
\begin{tikzpicture}

        \node (x) [circle] at (-1, 1.2) {$X$};
        \node (x) [circle] at (0.8,1.2) {$Y$};
        \node (x) [circle] at (2.3,1.2) {$Z$};
         \node [fill=blue!20, align=center, text height=1.5cm]{$P_{Y|X}$\\};
         \draw (-0.8,0.7) -- (-0.57, 0.7);
         \draw (-0.8,0.5) -- (-0.57, 0.5);
         \draw (-0.8,0.3) -- (-0.57, 0.3);
        \node (x4) [circle] at (-0.7,0.1) {$\vdots$};
        \draw (-0.8,-0.3) -- (-0.57, -0.3);
        \draw (-0.8,-0.5) -- (-0.57, -0.5);
        \draw (-0.8,-0.7) -- (-0.57, -0.7);

         \draw (0.57,0.7) -- (0.8, 0.7);
         \draw (0.57,0.5) -- (0.8, 0.5);
         \draw (0.57,0.3) -- (0.8, 0.3);
        \node (x4) [circle] at (0.8,0.1) {$\vdots$};
        \draw (0.57,-0.3) -- (0.8, -0.3);
        \draw (0.57,-0.5) -- (0.8, -0.5);
        \draw (0.57,-0.7) -- (0.8, -0.7);

        \node (z1) [circle] at (2.3,0.7) {$1$};
        \node (z2) [circle] at (2.3,-0.3) {e};
        \draw (0.8, 0.7) -- (1.95, 0.7);
        \draw (0.8, 0.7) -- (1.95, -0.3);
        \node [circle] at (1.4,0.88) {\footnotesize{$\delta$}};
        \node [circle, rotate=-45] at (1.4,0.35) {\footnotesize{$1-\delta$}};
         \draw (0.8, 0.5) -- (1.95,-0.3);
         \draw (0.8, 0.3) -- (1.95,-0.3);
        %\node (x4) [circle] at (0.8,0.1) {$\vdots$};

        \draw (0.8, -0.3) -- (1.95,-0.3);
        \draw (0.8, -0.5) -- (1.95,-0.3);
        \draw (0.8, -0.7) -- (1.95,-0.3);
\end{tikzpicture}
\caption{\small{The privacy filter associated with \eqref{filter_derivative} and \eqref{filter_derivative2} with $k=1$. We have $P_{Z|Y}(\cdot|1)=\sBer(\delta)$ and $P_{Z|Y}(\cdot|y)=\sBer(0)$ for $y\in\{2, 3, \dots, n\}$.}} \label{fig:Derivate_gEpsilon}
\end{figure}

We clearly have $P_Z(k)=\delta P_Y(k)$ and $P_Z(\text{e})=1-\delta P_Y(k)$, and hence
$$P_{X|Z}(x|k)=\frac{P_{XZ}(x,k)}{\delta P_Y(k)}=\frac{P_{XYZ}(x,k, k)}{\delta P_Y(k)}=\frac{\delta P_{XY}(x,k)}{\delta P_Y(k)}=P_{X|Y}(x|k),$$
and also,
\begin{eqnarray*}
  P_{X|Z}(x|\text{e})&=&\frac{P_{XZ}(x,\text{e})}{1-\delta P_Y(k)}= \frac{\sum_{y}P_{XYZ}(x,y, \text{e})}{1-\delta P_Y(k)}\\
   &=& \frac{\sum_{y\neq k}P_{XYZ}(x,y)+(1-\delta)P_{XY}(x, k)}{1-\delta P_Y(k)}=\frac{P_X(x)-\delta P_{XY}(x, k)}{1-\delta P_Y(k)}.
\end{eqnarray*}
It, therefore, follows that for $k\in\{1,2,\dots, n\}$
$$H(X|Z=k)=H(X|Y=k),$$ and
$$H(X|Z=\text{e})=H\left(\frac{P_X(1)-\delta P_{XY}(1, k)}{1-\delta P_Y(k)}, \dots, \frac{P_X(m)-\delta P_{XY}(m, k)}{1-\delta P_Y(k)}\right)=:\mathcal{h}_X(\delta).$$
We then write
$$I(X;Z)=H(X)-H(X|Z)=H(X)-\delta P_{Y}(k)H(X|Y=k)-(1-\delta P_{Y}(k))\mathcal{h}_X(\delta),$$ and hence,
$$\frac{\text{d}}{\text{d}\delta}I(X;Z)=-P_{Y}(k)H(X|Y=k)+P_{Y}(k)\mathcal{h}_X(\delta)-(1-\delta P_{Y}(k))\mathcal{h}'_X(\delta),$$
where
$$\mathcal{h}'_X(\delta)=\frac{\text{d}}{\text{d}\delta}\mathcal{h}_X(\delta)=-\sum_{x=1}^m\frac{P_X(x)P_Y(k)-P_{XY}(x, k)}{[1-\delta P_{Y}(k)]^2}\log\left(\frac{P_X(x)-\delta P_{XY}(x,y)}{1-\delta P_{Y}(k)}\right).$$

Using the first-order approximation of mutual information for $\delta=0$, we can write
\begin{eqnarray}
    I(X;Z)&=&\frac{\text{d}}{\text{d}\delta}I(X;Z)|_{\delta=0}\delta+o(\delta)\nonumber\\
&=& \delta\left[\sum_{x=1}^mP_{XY}(x,k)\log\left(\frac{P_{XY}(x,k)}{P_X(x)P_Y(k)}\right)\right]+o(\delta)\nonumber\\
&=&\delta P_Y(k)D(P_{X|Y}(\cdot|k)||P_X(\cdot))+o(\delta).\label{IXZ_firstOrder}
\end{eqnarray}
Similarly, we can write
\begin{eqnarray*}
I(Y;Z)&=& h(Z)-\sum_{y=1}^nP_Y(y)h(Z|Y=y)=h(Z)-P_Y(k)h(\delta)=h(\delta P_Y(k))-P_Y(k)h(\delta)\\
&=&-\delta P_Y(k)\log(P_Y(k))-\Psi(1-\delta P_Y(k))+P_Y(k)\Psi(1-\delta),
\end{eqnarray*}
where $\Psi(x):=x\log x$ which yields
$$\frac{\text{d}}{\text{d}\delta}I(Y;Z)=-\Psi(P_Y(k))+P_Y(k)\log\left(\frac{1-\delta P_Y(k)}{1-\delta}\right).$$
From the above, we obtain
\begin{eqnarray}
    I(Y;Z)&=&\frac{\text{d}}{\text{d}\delta}I(Y;Z)|_{\delta=0}\delta+o(\delta)\nonumber\\
&=& -\delta \Psi(P_Y(k))+o(\delta).\label{IYZ_firstOrder}
\end{eqnarray}
Clearly from \eqref{IXZ_firstOrder}, in order for the filter $P_{Z|Y}$ specified in \eqref{filter_derivative} and \eqref{filter_derivative2} to belong to $\D_{\eps}(P_{XY})$, we must have
$$\frac{\eps}{\delta}=P_Y(k)D(P_{X|Y}(\cdot|k)||P_X(\cdot))+\frac{o(\delta)}{\delta},$$ and hence from \eqref{IYZ_firstOrder}, we have
$$I(Y;Z)=\frac{-\Psi(P_Y(k))}{P_Y(k)D(P_{X|Y}(\cdot|k)||P_X(\cdot))}\eps+o(\delta).$$ This immediately implies that
\begin{equation}\label{gprime}
    g'_0(X;Y)=\lim_{\eps\downarrow 0}\frac{g_{\eps}(X;Y)}{\eps}\geq \frac{-\Psi(P_Y(k))}{P_Y(k)D(P_{X|Y}(\cdot|k)||P_X(\cdot))}=\frac{-\log(P_Y(k))}{D\left(P_{X|Y}(\cdot|k)||P_X(\cdot)\right)},
\end{equation}
%$$\lim_{\eps\downarrow 0}\frac{g_{\eps}(X;Y)}{\eps}\geq \lim_{\delta\downarrow 0}\frac{d_Y}{d_X}+\frac{o(\delta)}{\delta}\frac{\delta}{\eps}=\frac{d_Y}{d_X}.$$
where we have used the assumption $g_0(X,Y)=0$ in the first equality.
\section{Completion of Proof of Theorem~\ref{Theorem_Reverse_Direction_BISO}}\label{Appendix_BISO_Uniform}

%\begin{eqnarray*}
%% \nonumber to remove numbering (before each equation)
%  D(P_{X|Y}(\cdot|0)||P_X(\cdot)) &=& \sum_{x\in\X} \beta_x \log\left(\frac{\beta_x}{p\alpha_x+(1-p)\beta_x}\right)\\
%   &=&  \sum_{x\in\X} \alpha_x \log\left(\frac{\alpha_x}{p\beta_x+(1-p)\alpha_x}\right)=D(1-p).
%\end{eqnarray*}
%The equality \eqref{ratio3} can then be written as
%\begin{equation}\label{equality_to_be_shown_uniquely}
%  \frac{-\log(p)}{D(p)}=\frac{-\log(1-p)}{D(1-p)}.
%\end{equation}
To prove that the equality \eqref{ratio3} has only one solution $p=\frac{1}{2}$, we first show the following lemma.

\begin{lemma}\label{Lemma_Unique_sol_BISO}
Let $P$ and $Q$ be two distributions over $\X=\{\pm 1, \pm 2, \dots, \pm k\}$ which satisfy $P(x)=Q(-x)$. Let $R_{\lambda}:=\lambda P+(1-\lambda)Q$ for $\lambda\in (0, 1)$. Then
\begin{equation}\label{Lemma_BISO_Appendix_Equation1}
\frac{D(P||R_{1-\lambda})}{D(P||R_{\lambda})}<\frac{\log(1-\lambda)}{\log(\lambda)},
\end{equation}
for $\lambda\in(0, \frac{1}{2})$ and
\begin{equation}\label{Lemma_BISO_Appendix_Equation2}
\frac{D(P||R_{1-\lambda})}{D(P||R_{\lambda})}>\frac{\log(1-\lambda)}{\log(\lambda)},
\end{equation}
for $\lambda\in(\frac{1}{2}, 1)$.
\end{lemma}
Note that it is easy to see that the map $\lambda\mapsto D(P||R_{\lambda})$ is convex and strictly decreasing and hence $D(P||R_{\lambda})>D(P||R_{1-\lambda})$ when $\lambda\in(0, \frac{1}{2})$ and $D(P||R_{\lambda})<D(P||R_{1-\lambda})$ when $\lambda\in(\frac{1}{2}, 1)$. Inequality \eqref{Lemma_BISO_Appendix_Equation1} and \eqref{Lemma_BISO_Appendix_Equation2} strengthen these monotonic behavior and show that $D(P||R_{\lambda})>\frac{\log(\lambda)}{\log(1-\lambda)}D(P||R_{1-\lambda})$ and $D(P||R_{\lambda})<\frac{\log(\lambda)}{\log(1-\lambda)}D(P||R_{1-\lambda})$ for $\lambda\in(0, \frac{1}{2})$ and $\lambda\in(\frac{1}{2}, 1)$, respectively.
\begin{proof}
Without loss of generality, we can assume that $P(x)>0$ for all $x\in \X$. Let $\X_+:=\{x\in \X| P(X)> P(-x)\}$, $\X_-:=\{x\in \X| P(X)< P(-x)\}$ and $\X_0:=\{x\in \X| P(X)= P(-x)\}$.
We notice that when $x\in\X_+$, then $-x\in\X_-$, and hence $|\X_+|=|\X_-|=m$ for a $0<m\leq k$. After relabelling if needed, we can therefore assume that $\X_+=\{1, 2,\dots, m\}$ and $\X_-=\{-m,\dots, -2, -1\}$.
We can write
\begin{eqnarray*}
  D(P||R_{\lambda})&=&\sum_{x=-k}^k \log\left(\frac{P(x)}{\lambda P(x)+(1-\lambda)Q(x)}\right)=\sum_{x=-k}^k \log\left(\frac{P(x)}{\lambda P(x)+(1-\lambda)P(-x)}\right)\\
   &\stackrel{(a)}{=}& \sum_{x=1}^m \left[P(x)\log\left(\frac{P(x)}{\lambda P(x)+(1-\lambda)P(-x)}\right)
+P(-x)\log\left(\frac{P(-x)}{\lambda P(-x)+(1-\lambda)P(x)}\right)\right]\\
&\stackrel{(b)}{=}& \sum_{x=1}^m \left[P(x)\log\left(\frac{1}{\lambda+(1-\lambda)\zeta_x}\right)
+P(x)\zeta_x\log\left(\frac{1}{\lambda+\frac{(1-\lambda)}{\zeta_x}}\right)\right]\\
&\stackrel{(c)}{=}&\sum_{x=1}^m P(x)\Upsilon(\lambda,\zeta_x)\,\log\left(\frac{1}{\lambda}\right),
\end{eqnarray*}
where $(a)$ follows from the fact that for  $x\in \X_0$, $\log\left(\frac{P(x)}{R_{\lambda}(x)}\right)=0$ for any $\lambda\in (0, 1)$, and in $(b)$ and $(c)$ we introduced $\zeta_x:=\frac{P(-x)}{P(x)}$ and
$$\Upsilon(\lambda,\zeta):=\frac1{\log\left(\frac{1}{\lambda}\right)}\left(\log\left(\frac1{\lambda +(1-\lambda)\zeta}\right)+\zeta\log\left(\frac1{\lambda+\frac{(1-\lambda)}{\zeta}}\right)\right).$$
Similarly, we can write
\begin{eqnarray*}
  D(P||R_{1-\lambda})&=&\sum_{x=-k}^k \log\left(\frac{P(x)}{(1-\lambda) P(x)+\lambda Q(x)}\right)=\sum_{x=-k}^k \log\left(\frac{P(x)}{(1-\lambda) P(x)+\lambda P(-x)}\right)\\
   &=& \sum_{x=1}^m \left[P(x)\log\left(\frac{P(x)}{ (1-\lambda)P(x)+\lambda P(-x)}\right)
+P(-x)\log\left(\frac{P(-x)}{(1-\lambda) P(-x)+\lambda P(x)}\right)\right]\\
&=& \sum_{x=1}^m \left[P(x)\log\left(\frac{1}{1-\lambda+\lambda\zeta_x}\right)
+P(x)\zeta_x\log\left(\frac{1}{1-\lambda+\frac{\lambda}{\zeta_x}}\right)\right]\\
&=&\sum_{x=1}^m P(x)\Upsilon(1-\lambda,\zeta_x)\,\log\left(\frac{1}{1-\lambda}\right),
\end{eqnarray*}
which implies that
$$\frac{D(P||R_{\lambda})}{-\log(\lambda)}-\frac{D(P||R_{1-\lambda})}{-\log(1-\lambda)}=
\sum_{x=1}^m P(x)\left[\Upsilon(\lambda, \zeta_x)-\Upsilon(1-\lambda, \zeta_x)\right].$$
Hence, in order to show \eqref{Lemma_BISO_Appendix_Equation1}, it suffices to verify that
\begin{equation}\label{Phi_Definition}
  \Phi(\lambda,\zeta):=\Upsilon(\lambda, \zeta)-\Upsilon(1-\lambda, \zeta)>0,
\end{equation}
for any $\lambda\in(0, \frac{1}{2})$ and $\zeta\in (1, \infty)$.
Since $\log(\lambda)\log(1-\lambda)$ is always positive for $\lambda\in (0, \frac{1}{2})$, it suffices to show that
\begin{equation}\label{h_definition}
  h(\zeta):=\Phi(\lambda,\zeta)\log(1-\lambda)\log(\lambda)>0,
\end{equation}
for $\lambda\in(0, \frac{1}{2})$ and $\zeta\in (1, \infty)$.
We have
\begin{equation}\label{Second_derivatie}
  h''(\zeta)=A(\lambda, \zeta)B(\lambda, \zeta),
\end{equation}
where
$$A(\lambda, \zeta):=\frac{1 + \zeta}{(1-\lambda+ \lambda \zeta)^2 (\lambda + (1- \lambda) \zeta)^2 \zeta} ,$$
and
$$B(\lambda, \zeta):=\lambda^2 (1 + \lambda(\lambda-2)(\zeta-1)^2 +\zeta(\zeta-1)) \log(\lambda)-(1- \lambda)^2 (\lambda^2 (\zeta-1)^2 + \zeta) \log(1 - \lambda).$$
We have
$$\frac{\partial^2}{\partial\zeta^2}B(\lambda, \zeta)=2\lambda^2(1-\lambda)^2\log\left(\frac{\lambda}{1-\lambda}\right)<0,$$
because $\lambda\in (0, \frac{1}{2})$ and hence $\lambda<1-\lambda$. This implies that the map $\zeta\mapsto B(\lambda, \zeta)$ is concave for any $\lambda\in (0, \frac{1}{2})$ and $\zeta\in (1, \infty)$. Moreover, since $\zeta\mapsto B(\lambda, \zeta)$ is a quadratic polynomial with negative leading coefficient, it is clear that $\lim_{\zeta\to \infty}B(\lambda, \zeta)=-\infty$.
Consider now $g(\lambda):=B(\lambda, 1)=\lambda^2\log(\lambda)-(1-\lambda)^2\log(1-\lambda)$. We have $\lim_{\lambda\to 0}g(\lambda)=g(\frac{1}{2})=0$ and $g''(\lambda)=2\log\left(\frac{\lambda}{1-\lambda}\right)<0$ for $\lambda\in (0, \frac{1}{2})$. It implies that $\lambda\mapsto g(\lambda)$ is concave over $(0, \frac{1}{2})$ and hence $g(\lambda)>0$ over $(0, \frac{1}{2})$ which implies that $B(\lambda, 1)>0$. This together with the fact that $\zeta\mapsto B(\lambda, \zeta)$ is concave and it approaches to $-\infty$ as $\zeta\to \infty$ imply that there exists a real number $c=c(\lambda)>1$ such that $B(\lambda, \zeta)>0$ for all $\zeta\in (1, c)$ and $B(\lambda, \zeta)<0$ for all $\zeta\in (c, \infty)$. Since $A(\lambda, \zeta)>0$, it follows from \eqref{Second_derivatie} that $\zeta\mapsto h(\zeta)$ is convex over $(1, c)$ and concave over $(c, \infty)$. Since $h(1)=h'(1)=0$ and $\lim_{\zeta\to \infty}h(\zeta)=\infty$, we can conclude that $h(\zeta)>0$ over $(1, \infty)$. That is, $\Phi(\lambda, \zeta)>0$ and thus
$\Upsilon(\lambda, \zeta)-\Upsilon(1-\lambda, \zeta)>0$, for $\lambda\in(0, \frac{1}{2})$ and $\zeta\in (1, \infty)$.

The inequality \eqref{Lemma_BISO_Appendix_Equation2} can be proved by \eqref{Lemma_BISO_Appendix_Equation1} and  switching $\lambda$ to $1-\lambda$.
\end{proof}
Letting $P(\cdot)=P_{X|Y}(\cdot|1)$ and $Q(\cdot)=P_{X|Y}(\cdot|0)$ and $\lambda=\Pr(Y=1)=p$, we have $R_{p}(x)=P_X(x)=pP(x)+(1-p)Q(x)$ and $R_{1-p}=P_X(-x)=(1-p)P(x)+pQ(x)$. Since $D(P_{X|Y}(\cdot|0)||P_X(\cdot))=D(P||R_{1-p})$, we can conclude from Lemma~\ref{Lemma_Unique_sol_BISO} that
$$\frac{D(P_{X|Y}(\cdot|0)||P_X(\cdot))}{-\log(1-p)}<\frac{D(P_{X|Y}(\cdot|1)||P_X(\cdot))}{-\log(p)},$$
over $p\in (0, \frac{1}{2})$ and
$$\frac{D(P_{X|Y}(\cdot|0)||P_X(\cdot))}{-\log(1-p)}>\frac{D(P_{X|Y}(\cdot|1)||P_X(\cdot))}{-\log(p)},$$
over $p\in (\frac{1}{2}, 1)$, and hence equation \eqref{ratio3} has only solution $p=\frac{1}{2}$.

\section{Proof of Theorems \ref{Thm:gMtog} and \ref{Thm:gContinuous}}
\label{Appendix:ProofsContinuous}

The proof of Theorem \ref{Thm:gContinuous} does not depend on the proof of Theorem \ref{Thm:gMtog}, so, there is no harm in proving the former theorem first. The following version of the data-processing inequality will be required.
%Recall the log-sum inequality for functions on $\R$.

%\begin{theorem}[REF]
%Let $f,g:\R\to\R$ be non-negative functions such that $x\mapsto f(x)\log(f(x)/g(x))$ is integrable and $0<\int f(x)\text{d} x,\int g(x)\text{d} x<\infty$. Then
%\begin{equation*}
%\int f(x) \log\frac{f(x)}{g(x)} \text{d} x \geq \left(\int f(x) \text{d} x\right) \log \frac{\int f(x) \text{d} x}{\int g(x) \text{d} x}
%\end{equation*}
%with equality if and only if $f(x)=g(x)$ almost surely (w.r.t. the Lebesgue measure on $\R$).
%\end{theorem}

\begin{lemma}
\label{Thm:DataProcessingInq}
Let $X$ and $Y$ be absolutely continuous random variables such that $X$, $Y$ and $(X,Y)$ have finite differential entropies. If $V$ is an absolutely continuous random variable independent of $X$ and $Y$, then
\begin{equation*}
I(X;Y+V) \leq I(X;Y)
\end{equation*}
with equality if and only if $X$ and $Y$ are independent.
\end{lemma}

%Recall that the assumption $f_V(x)>0$ for all $x\in\R$ is satisfied, in particular, by the Gaussian density $f(x)=\dfrac{1}{\sqrt{2\pi}}\exp(-x^2/2)$.

\begin{proof}
Since $X\markov Y\markov (Y+V)$, the data processing inequality implies that $I(X;Y+V)\leq I(X;Y)$. It therefore suffices to show that this inequality is tight if and only $X$ and $Y$ are independent. It is known that data processing inequality is tight if and only if $X\markov (Y+V)\markov Y$. This is equivalent to saying that for any measurable set $A\subset \R$ and for $P_{Y+V}$ almost all $z$, $\Pr(X\in A|Y+V=z, Y=y)=\Pr(X\in A|Y+V=z)$. On the other hand, due to the independence of $V$ and $(X,Y)$, we  have $\Pr(X\in A|Y+V=z, Y=y)=\Pr(X\in A|Y=z-v)$. Hence, the equality holds if and only if $\Pr(X\in A|Y+V=z)=\Pr(X\in A|Y=z-v)$ which implies that $X$ and $Y$ must be independent.
\end{proof}

%Recall the following theorem from \cite{God04}. Even though it is stated for complex random variables, the same assertion holds for real random variables (see \cite[Sec. III]{God04}).

%\begin{theorem}[Thm. 1, \cite{God04}]
%\label{Thm:Hero}
%Let $(f_n)_{n\geq1}$ be a sequence of probability density functions in $\C^P$ such that $f_n \to f_0$ pointwise for some density function $f_0$. If
%\begin{equation*}
%\sup_{n\geq0,x\in\C^P} f_n(x) < \infty \quad\quad\quad\textnormal{and}\quad\quad\quad \sup_{n\geq0} \int_{\C^P} ||x||^p f_n(x) \text{d} x <\infty
%\end{equation*}
%for some $p>1$, then $h(f_n)\to h(f_0)$.
%\end{theorem}

\begin{lemma}
\label{Lemma:MonotonicContinuityI}
In the notation of Section~\ref{Subsection:gContinuous}, the function $\gamma\mapsto I(Y;Z_\gamma)$ is strictly-decreasing and continuous. Additionally, it satisfies
\begin{equation*}
I(Y;Z_\gamma) \leq \frac{1}{2} \log\left(1+\frac{\var(Y)}{\gamma^2}\right).
\end{equation*}
with equality if and only if $Y$ is Gaussian. In particular, $I(Y;Z_\gamma)\to 0$ as $\gamma\to\infty$.
\end{lemma}

\begin{proof}
Recall that, by assumption b), $\var(Y)$ is finite. The finiteness of the entropy of $Y$ follows from assumption, the corresponding statement for $Y+\gamma N$ follows from a routine application of the entropy power inequality \cite[Theorem 17.7.3]{Cover_Book} and the fact that $\var(Y+\gamma N)=\var(Y)+\gamma^2<\infty$, and for $(Y,Y+\gamma N)$ the same conclusion follows by the chain rule for differential entropy. The data processing inequality, as stated in Lemma \ref{Thm:DataProcessingInq}, implies
\begin{equation*}
I(Y;Z_{\gamma+\delta}) \leq I(Y;Y+\gamma N) = I(Y;Z_\gamma).
\end{equation*}
Clearly $Y$ and $Y+\gamma N$ are not independent, therefore the inequality is strict and thus $\gamma\mapsto I(Y,Z_\gamma)$ is strictly-decreasing.

Continuity will be studied for $\gamma=0$ and $\gamma>0$ separately. Recall that $h(\gamma N) = \frac{1}{2} \log (2\pi e \gamma^2)$. In particular, $\displaystyle \lim_{\gamma\to0} h(\gamma N) = -\infty$. The entropy power inequality shows then that $\displaystyle \lim_{\gamma\to0} I(Y;Y+\gamma N) = \infty$. This coincides with the convention $I(Y;Z_0)=I(Y;Y)=\infty$. For $\gamma>0$, let $(\gamma_n)_{n\geq1}$ be a sequence of positive numbers such that $\gamma_n\to\gamma$. Observe that
\begin{align*}
I(Y;Z_{\gamma_n}) &= h(Y+\gamma_n N) - h(\gamma_n N) = h(Y+\gamma_n N) - \frac{1}{2}\log(2\pi e\gamma_n^2).
\end{align*}
Since $\displaystyle \lim_{n\to\infty} \frac{1}{2} \log(2\pi e\gamma_n^2) = \frac{1}{2} \log(2\pi e\gamma^2)$, we only have to show that $h(Y+\gamma_n N)\to h(Y+\gamma N)$ as $n\to\infty$ to establish the continuity at $\gamma$. This, in fact, follows from de Bruijn's identity (cf., \cite[Theorem 17.7.2]{Cover_Book}).

%Let $\gamma_0=\gamma$. Since $\gamma_n>0$ and $\gamma_n\to\gamma_0$, we have that $0<\displaystyle \sup_{n\geq0} \gamma_n <\infty$. This implies that
%\begin{equation*}
%\sup_{n\geq 0,z\in\R} f_{Y+\gamma_n N}(z) < \infty.
%\end{equation*}
%Also, for all $n\geq0$, $\E[(Y+\gamma_n N)^2] = \E[Y^2] + \gamma_n^2$ and thus
%\begin{equation*}
%\sup_{n\geq 0} \E[(Y+\gamma_n N)^2] \leq \E[Y^2]  + \sup_{n\geq0} \gamma_n^2 < \infty.
%\end{equation*}
%A routine application of the dominated convergence theorem shows that $f_{Y+\gamma_n N} \to f_{Y+\gamma N}$ pointwise and therefore, by Theorem \ref{Thm:Hero}, $h(Y+\gamma_n N) \to h(Y+\gamma N)$ as $n\to\infty$.

Since the channel from $Y$ to $Z_\gamma$ is an additive Gaussian noise channel, we have $\displaystyle I(Y;Z_\gamma)\leq\frac{1}{2} \log\left(1+\frac{\var(Y)}{\gamma^2}\right)$ with equality if and only if $Y$ is Gaussian. The claimed limit as $\gamma\to 0$ is clear.
\end{proof}

\begin{lemma}\label{Lemma:DecreasingIXZ}
The function $\gamma\mapsto I(X;Z_\gamma)$ is strictly-decreasing and continuous. Moreover, $I(X;Z_\gamma)\to0$ when $\gamma\to\infty$.
\end{lemma}

\begin{proof}
The proof of the strictly-decreasing behavior of $\gamma\mapsto I(X;Z_\gamma)$ is proved as in the previous lemma.

To prove continuity, let $\gamma\geq0$ be fixed. Let $(\gamma_n)_{n\geq1}$ be any sequence of positive numbers converging to $\gamma$. First suppose that $\gamma>0$. Observe that
\begin{equation*}
I(X;Z_{\gamma_n}) = h(Y+\gamma_n N) - h(Y+\gamma_n N|X)
\end{equation*}
for all $n\geq1$. As shown in Lemma~\ref{Lemma:MonotonicContinuityI}, $h(Y+\gamma_n N)\to h(Y+\gamma N)$ as $n\to\infty$. Therefore, it is enough to show that $h(Y+\gamma_n N|X) \to h(Y+\gamma N|X)$ as $n\to\infty$.  Note that by de Bruijn's identity, we have  $h(Y+\gamma_n N|X=x)\to h(Y+\gamma N|X=x)$ as $n\to\infty$ for all $x\in \R$. Note also that since
$$h(Z_{\gamma_n}|X=x)\leq \frac{1}{2}\log\left(2\pi e\var(Z_{\gamma_n}|x)\right),$$ we can write
$$h(Z_{\gamma_n}|X)\leq \E\left[\frac{1}{2}\log(2\pi e\var(Z_{\gamma_n}|X))\right]\leq \frac{1}{2}\log\left(2\pi e\E[\var(Z_{\gamma_n}|X)]\right),$$
and hence we can apply dominated convergence theorem to show that $h(Y+\gamma_n N|X)\to h(Y+\gamma N|X)$ as $n\to\infty$.

%By assumption $||f_{X,Y}|| = \sup_{x,y} f_{X,Y}(x,y) < \infty$, thus, for all $(x,z)\in\R^2$ and all $n\geq0$,
%\begin{equation*}
%f_{X,Y+\gamma_n N}(x,z) = \int_\R f_{X,Y}(x,y) f_{\gamma_n N} (z-y) \text{d} y \leq ||f_{X,Y}||.
%\end{equation*}
%In particular,
%\begin{equation*}
%\sup_{n\geq0,(x,z)\in\R^2} f_{X,Y+\gamma_n N}(x,z) <\infty.
%\end{equation*}
%Also, for all $n\geq0$,
%\begin{equation*}
%\E\left[X^2+(Y+\gamma_n N)^2\right] = \E\left[X^2+Y^2\right] + \gamma_n^2.
%\end{equation*}
%By assumption b), the first summand in the RHS of the last equation is finite, and thus
%\begin{equation*}
%\sup_{n\geq0} \E\left[X^2+(Y+\gamma_n N)^2\right] < \infty.
%\end{equation*}
%A routine application of the dominated convergence theorem implies that\footnote{When $\gamma=0$ a direct computation of the limits should be perform instead (using the continuity of $(x,y)\mapsto f_{X,Y}(x,y)$).} $f_{X,Y+\gamma_n N} \to f_{X,Y+\gamma N}$ as $n\to\infty$ and therefore, by Theorem \ref{Thm:Hero}, $h(X,Y+\gamma_n N)\to h(X,Y+\gamma N)$ as $n\to\infty$.
To prove the continuity at $\gamma=0$, we first note that Linder and Zamir \cite[Page~2028]{Linder} showed that $h(Y+\gamma_nN|X=x)\to h(Y|X=x)$ as $n\to \infty$, then, as before, by dominated convergence theorem we can show that $h(Y+\gamma_nN|X)\to h(Y|X)$. Similarly \cite{Linder} implies that  $h(Y+\gamma_nN)\to h(Y)$. This concludes the proof of the continuity of $\gamma\mapsto I(X; Z_{\gamma})$.

Furthermore, by the data processing inequality and previous lemma,
\begin{equation*}
0 \leq I(X;Z_\gamma) \leq I(Y;Z_\gamma) \leq \frac{1}{2} \log\left(1+\frac{\var(Y)}{\gamma^2}\right),
\end{equation*}
 and hence we conclude that $\displaystyle \lim_{\gamma\to\infty} I(X;Z_\gamma) = 0$.
\end{proof}

\begin{proof}[{\bf Proof of Theorem \ref{Thm:gContinuous}}]
The nonnegativity of $g_\eps(X;Y)$ follows directly from definition.

By Lemma~\ref{Lemma:DecreasingIXZ}, for every $0<\eps\leq I(X;Y)$ there exists a unique $\gamma_\eps\in[0,\infty)$ such that $I(X;Z_{\gamma_\eps})=\eps$, so $g_\eps(X;Y) = I(Y;Z_{\gamma_\eps})$. Moreover, $\eps\mapsto\gamma_\eps$ is strictly decreasing. Since $\gamma\mapsto I(Y;Z_\gamma)$ is strictly-decreasing, we conclude that $\eps\mapsto g_\eps(X;Y)$ is strictly increasing.

The fact that  $\eps\mapsto\gamma_\eps$ is strictly decreasing, also implies that $\gamma_\eps \to \infty$ as $\eps\to0$. In particular,
\begin{equation*}
\lim_{\eps\to0} g_\eps(X;Y) = \lim_{\eps\to0} I(Y;Z_{\gamma_\eps}) = \lim_{\gamma_\eps\to\infty} I(Y;Z_{\gamma_\eps}) = \lim_{\gamma\to\infty} I(Y;Z_\gamma) = 0.
\end{equation*}
By the data processing inequality we have that $I(X;Z_\gamma)\leq I(X;Y)$ for all $\gamma\geq0$, i.e., any filter satisfies the privacy constraint for $\eps=I(X;Y)$. Thus, $\displaystyle g_{I(X;Y)}(X;Y) \geq I(Y;Y)=\infty$.
\end{proof}

%\noindent------

 In order to prove Theorem \ref{Thm:gMtog}, we first recall the following theorem by R\'enyi \cite{Ren59}.
\begin{theorem}[\cite{Ren59}]
If $U$ is an absolutely continuous random variable with density $f_U(x)$ and if $H(\lfloor U\rfloor)<\infty$, then
\begin{equation*}
\lim_{n\to\infty} H(n^{-1}\lfloor nU\rfloor) - \log(n) = - \int_\R f_U(x) \log f_U(x) \text{d} x,
\end{equation*}
provided that the integral on the right hand side exists.
\end{theorem}

We will need the following consequence of the previous theorem.

\begin{lemma}
\label{Lemma:QuantizationDecreasing}
If $U$ is an absolutely continuous random variable with density $f_U(x)$ and if $H(\lfloor U\rfloor)<\infty$, then $H(\Q_M(U)) - M \geq H(\Q_{M+1}(U)) - (M+1)$ for all $M\geq 1$ and
\begin{equation*}
\lim_{n\to\infty} H(\Q_M(U)) - M = - \int_\R f_U(x) \log f_U(x) \text{d} x,
\end{equation*}
provided that the integral on the right hand side exists.
\end{lemma}

The previous lemma follows from the fact that $\Q_{M+1}(U)$ is constructed by refining the quantization partition for $\Q_M(U)$.

\begin{lemma}
\label{Lemma:ContinuityMutualInformations}
For any $\gamma\geq0$,
\begin{equation*}
\lim_{M\to\infty} I(X;Z^M_\gamma) = I(X;Z_\gamma)\quad\quad\quad\textnormal{and}\quad\quad\quad\lim_{M\to\infty} I(Y;Z^M_\gamma) = I(Y;Z_\gamma).
\end{equation*}
\end{lemma}

\begin{proof}
Observe that
\begin{align*}
I(X;Z^M_\gamma) &= I(X;\Q_M(Y+\gamma N))\\
&= H(\Q_M(Y+\gamma N)) - H(\Q_M(Y+\gamma N) | X)\\
&= [H(\Q_M(Y+\gamma N))-M]-\int_\R f_X(x) [H(\Q_M(Y+\gamma N)|X=x)-M] \text{d} x.
\end{align*}
By the previous lemma, the integrand is decreasing in $M$, and thus we can take the limit with respect to $M$ inside the integral. Thus,
\begin{equation*}
\lim_{M\to\infty} I(X;Z^M_\gamma) = h(Y+\gamma N) - h(Y+\gamma N|X) = I(X;Z_\gamma).
\end{equation*}
The proof for $I(Y; Z^M_{\gamma})$ is analogous.
\end{proof}
\begin{lemma}
Fix $M\in\N$. Assume that $f_Y(y) \leq C |y|^{-p}$ for some positive constant $C$ and $p>1$. For integer $k$ and $\gamma\geq0$, let
\begin{equation*}
p_{k,\gamma} := \Pr\left(\Q_M(Y+\gamma N)=\frac{k}{2^M}\right).
\end{equation*}
Then
\begin{equation*}
p_{k,\gamma} \leq \frac{C2^{(p-1)M+p}}{k^p} + 1_{\{\gamma>0\}} \frac{\gamma 2^{M+1}}{k\sqrt{2\pi}} e^{-k^2/2^{2M+3}\gamma^2}.
\end{equation*}
\end{lemma}

\begin{proof}
The case $\gamma=0$ is trivial, so we assume that $\gamma>0$. For notational simplicity, let $r_a=\frac{a}{2^M}$ for all $a\in\mathbb{Z}$. Assume that $k\geq0$. Observe that
\begin{align*}
p_{k,\gamma} &= \int_{-\infty}^\infty \int_{-\infty}^\infty f_{\gamma N}(n) f_Y(y) 1_{\left[r_k,r_{k+1}\right)}(y+n) \text{d} y\text{d} n\\
&= \int_{-\infty}^\infty \frac{e^{-n^2/2\gamma^2}}{\sqrt{2\pi\gamma^2}} \Pr\left(Y\in\left[r_k,r_{k+1}\right)-n\right) \text{d} n.
\end{align*}
We will estimate the above integral by breaking it up into two pieces.

First, we consider
\begin{align*}
\int\limits_{-\infty}^{\frac{r_k}{2}} \frac{e^{-n^2/2\gamma^2}}{\sqrt{2\pi\gamma^2}} \Pr\left(Y\in\left[r_k,r_{k+1}\right)-n\right) \text{d} n.
\end{align*}
When $n\leq \frac{r_k}{2}$, then $r_k-n\geq r_k/2$. By the assumption on the density of $Y$,
\begin{align*}
\Pr\left(Y\in\left[r_k,r_{k+1}\right)-n\right) &\leq \frac{C}{2^M} \left(\frac{r_k}{2}\right)^{-p}.
\end{align*}
(The previous estimate is the only contribution when $\gamma=0$.) Therefore,
\begin{align*}
\int\limits_{-\infty}^{\frac{r_k}{2}} \frac{e^{-n^2/2\gamma^2}}{\sqrt{2\pi\gamma^2}} \Pr\left(Y\in\left[r_k,r_{k+1}\right)-n\right) \text{d} n &\leq \frac{C}{2^M} \left(\frac{r_k}{2}\right)^{-p} \int\limits_{-\infty}^{\frac{r_k}{2}} \frac{e^{-n^2/2\gamma^2}}{\sqrt{2\pi\gamma^2}}  \text{d} n\\
&\leq \frac{C2^{(p-1)M+p}}{k^p}.
\end{align*}

Using the trivial bound $\Pr\left(Y\in\left[r_k,r_{k+1}\right)-n\right)\leq1$ and well known estimates for the error function, we obtain that
\begin{align*}
\int\limits_{\frac{r_k}{2}}^\infty \frac{e^{-n^2/2\gamma^2}}{\sqrt{2\pi\gamma^2}} \Pr\left(Y\in\left[r_k,r_{k+1}\right)-n\right) \text{d} n %&\leq \int\limits_{\frac{\hat{k}}{2}}^\infty \frac{e^{-n^2/2\gamma^2}}{\sqrt{2\pi\gamma^2}} \text{d} n\\
&< \frac{1}{\sqrt{2\pi}} \frac{2\gamma}{r_k} e^{-r_k^2/8\gamma^2}\\
&= \frac{\gamma 2^{M+1}}{k\sqrt{2\pi}} e^{-k^2/2^{2M+3}\gamma^2}.
\end{align*}

Therefore,
\begin{equation*}
p_{k,\gamma} \leq \frac{C2^{(p-1)M+p}}{k^p} + \frac{\gamma 2^{M+1}}{k\sqrt{2\pi}} e^{-k^2/2^{2M+3}\gamma^2}.
\end{equation*}
The proof for $k<0$ is completely analogous.
\end{proof}

\begin{lemma}
Fix $M\in\N$. Assume that $f_Y(y) \leq C |y|^{-p}$ for some positive constant $C$ and $p>1$. The mapping $\gamma\mapsto H(\Q_M(Y+\gamma N))$ is continuous.
\end{lemma}

\begin{proof}
Let $(\gamma_n)_{n\geq1}$ be a sequence of non-negative real numbers converging to $\gamma_0$. First, we will prove continuity at $\gamma_0>0$. Without loss of generality, assume that $\gamma_n>0$ for all $n\in\N$. Define $\gamma_\ast=\inf\{\gamma_n | n\geq1\}$ and $\gamma^\ast=\sup\{\gamma_n | n\geq1\}$. Clearly $0<\gamma_\ast\leq\gamma^\ast<\infty$. Recall that
\begin{equation*}
p_{k,\gamma} = \int_\R \frac{e^{-z^2/2\gamma^2}}{\sqrt{2\pi\gamma^2}} \Pr\left(Y\in\left[\frac{k}{2^M},\frac{k+1}{2^M}\right)-z\right) \text{d} z.
\end{equation*}
Since, for all $n\in\N$ and $z\in\R$,
\begin{align*}
\frac{e^{-z^2/2\gamma_n^2}}{\sqrt{2\pi\gamma_n^2}} \Pr\left(Y\in\left[\frac{k}{2^M},\frac{k+1}{2^M}\right)-z\right) &\leq \frac{e^{-z^2/2(\gamma^\ast)^2}}{\sqrt{2\pi\gamma_\ast^2}},
\end{align*}
the dominated convergence theorem implies that
\begin{equation}
\label{eq:limprobs}
\lim_{n\to\infty} p_{k,\gamma_n} = p_{k,\gamma_0}.
\end{equation}

The previous lemma implies that for all $n\geq0$ and $|k|>0$,
\begin{equation*}
p_{k,\gamma_n} \leq \frac{C2^{(p-1)M+p}}{k^p} + \frac{\gamma_n 2^{M+1}}{k\sqrt{2\pi}} e^{-k^2/2^{2M+3}\gamma_n^2}.
\end{equation*}
Thus, for $k$ large enough, $\displaystyle p_{k,\gamma_n} \leq \frac{A}{k^p}$ for a suitable positive constant $A$ that does not depend on $n$. Since the function $x\mapsto -x\log(x)$ is increasing in $[0,1/2]$, there exists $K'>0$ such that for $|k|>K'$
\begin{equation*}
- p_{k,\gamma_n}\log(p_{k,\gamma_n}) \leq \frac{A}{k^p}\log(A^{-1}k^p).
\end{equation*}
Since $\displaystyle \sum_{|k|>K'} \frac{A}{k^p}\log(A^{-1}k^p) < \infty$, for any $\epsilon>0$ there exists $K_\epsilon$ such that
\begin{equation*}
\sum_{|k|>K_\epsilon} \frac{A}{k^p}\log(A^{-1}k^p) < \epsilon.
\end{equation*}
In particular, for all $n\geq0$,
\begin{align*}
H(\Q(Y+\gamma_n N)) - \sum_{|k|\leq K_\epsilon} - p_{k,\gamma_n}\log(p_{k,\gamma_n}) &= \sum_{|k|>K_\epsilon} - p_{k,\gamma_n}\log(p_{k,\gamma_n}) < \epsilon.
\end{align*}
Therefore, for all $n\geq1$,
\begin{align*}
&\left| H(\Q(Y+\gamma_n N)) - H(\Q(Y+\gamma_0 N)) \right|\\
&\leq \sum_{|k|>K_\epsilon} - p_{k,\gamma_n}\log(p_{k,\gamma_n}) + \left| \sum_{|k|\leq K_\epsilon} p_{k,\gamma_0}\log(p_{k,\gamma_0}) - p_{k,\gamma_n}\log(p_{k,\gamma_n}) \right| + \sum_{|k|>K_\epsilon} - p_{k,\gamma_0}\log(p_{k,\gamma_0})\\
&\leq \epsilon + \left| \sum_{|k|\leq K_\epsilon} p_{k,\gamma_0}\log(p_{k,\gamma_0}) - p_{k,\gamma_n}\log(p_{k,\gamma_n}) \right| + \epsilon.
\end{align*}
By continuity of the function $x\mapsto -x\log(x)$ on $[0,1]$ and equation (\ref{eq:limprobs}), we conclude that
\begin{equation*}
\limsup_{n\to\infty} \left| H(\Q(Y+\gamma_n N)) - H(\Q(Y+\gamma_0 N)) \right| \leq 3\epsilon.
\end{equation*}
Since $\epsilon$ is arbitrary,
\begin{equation*}
\lim_{n\to\infty} H(\Q(Y+\gamma_n N)) = H(\Q(Y+\gamma_0 N)),
\end{equation*}
as we wanted to prove.

To prove continuity at $\gamma_0=0$, observe that equation (\ref{eq:limprobs}) holds in this case as well. The rest is analogous to the case $\gamma_0>0$.
\end{proof}

\begin{lemma}
The functions $\gamma\mapsto I(X;Z_\gamma^M)$ and $\gamma\mapsto I(Y;Z_\gamma^M)$ are continuous for each $M\in \N$.
\end{lemma}

\begin{proof}
Since $H(\Q_M(Y+\gamma N)|Y=y)$ and $H(\Q_M(Y+\gamma N)|X=x)$ for $x, y\in\R$ are bounded by $M$, and $f_{Y|X}(y|x)$ satisfies assumption (b), the conclusion follows from the dominated convergence theorem.
\end{proof}

\begin{proof}[{\bf Proof of Theorem \ref{Thm:gMtog}}]
For every $M\in\N$, let $\Gamma_\epsilon^M:=\{\gamma\geq0 | I(X;Z_\gamma^M)\leq\epsilon\}$. The Markov chain $X\to Y\to Z_\gamma \to Z_\gamma^{M+1}\to Z_\gamma^M$ and the data processing inequality imply that
\begin{equation*}
I(X;Z_\gamma) \geq I(X;Z_{\gamma}^{M+1}) \geq I(X;Z_{\gamma}^M),
\end{equation*}
and, in particular,
\begin{equation*}
\epsilon = I(X;Z_{\gamma_\epsilon}) \geq I(X;Z_{\gamma_\epsilon}^{M+1}) \geq I(X;Z_{\gamma_\epsilon}^M),
\end{equation*}
where $\gamma_\epsilon$ is as defined in the proof of Theorem \ref{Thm:gContinuous}. This implies then that
\begin{equation}
\label{eq:gammaepsGamma}
\gamma_\epsilon \in \Gamma_\epsilon^{M+1} \subset \Gamma_\epsilon^M,
\end{equation}
and thus
\begin{equation*}
I(Y;Z_{\gamma_{\epsilon}}^M) \leq g_{\epsilon,M}(X;Y).
\end{equation*}
Taking limits in both sides, Lemma \ref{Lemma:ContinuityMutualInformations} implies
\begin{equation}
\label{eq:gMliminf}
g_{\epsilon}(X;Y) = I(Y;Z_{\gamma_{\epsilon}}) \leq \liminf_{M\to\infty} g_{\epsilon,M}(X;Y).
\end{equation}

Observe that
\begin{align}
\nonumber g_{\epsilon,M}(X;Y) &= \sup_{\gamma\in\Gamma_\epsilon^M} I(Y;Z_\gamma^M)\\
\nonumber &\leq \sup_{\gamma\in\Gamma_\epsilon^M} I(Y;Z_\gamma)\\
\label{eq:Inq} &= I(Y;Z_{\gamma_{\epsilon,min}^M}),
\end{align}
where inequality follows from Markovity and $\gamma_{\epsilon,\min}^M := \inf_{\Gamma_\epsilon^M}\gamma$. By equation (\ref{eq:gammaepsGamma}), $\gamma_\epsilon\in\Gamma_\epsilon^{M+1} \subset \Gamma_\epsilon^M$ and in particular $\gamma_{\epsilon,\min}^M \leq \gamma_{\epsilon,\min}^{M+1}\leq\gamma_\epsilon$. Thus, $\{\gamma_{\eps,\min}^M\}$ is an increasing sequence in $M$ and bounded from above and, hence, has a limit. Let $\displaystyle \gamma_{\epsilon,\min}=\lim_{M\to\infty} \gamma_{\epsilon,\min}^M$.  Clearly
\begin{equation}
\label{eq:MainThmleq}
\gamma_{\epsilon,\min} \leq \gamma_\epsilon.
\end{equation}

By the previous lemma we know that $I(X;Z_\gamma^M)$ is continuous, so $\Gamma_\epsilon^M$ is closed for all $M\in\N$. Thus, we have that $\gamma_{\epsilon,\min}^M = \min_{\Gamma_\epsilon^M}\gamma$ and in particular $\gamma_{\epsilon,\min}^M\in\Gamma_\epsilon^M$. By the inclusion $\Gamma_\epsilon^{M+1} \subset \Gamma_\epsilon^M$, we have then that $\gamma_{\epsilon,\min}^{M+n}\in\Gamma_\epsilon^M$ for all $n\in\N$. By closedness of $\Gamma_\epsilon^M$ we have then that $\gamma_{\epsilon,\min}\in\Gamma_\epsilon^M$ for all $M\in\N$. In particular,
\begin{equation*}
I(X;Z_{\gamma_{\epsilon,\min}}^M) \leq \epsilon,
\end{equation*}
for all $M\in\N$. By Lemma \ref{Lemma:ContinuityMutualInformations},
\begin{equation*}
I(X;Z_{\gamma_{\epsilon,\min}}) \leq \epsilon = I(X;Z_{\gamma_\epsilon}),
\end{equation*}
and by the monotonicity of $\gamma\mapsto I(X;Z_\gamma)$, we obtain that $\gamma_\epsilon \leq \gamma_{\epsilon,\min}$. Combining the previous inequality with (\ref{eq:MainThmleq}) we conclude that $\gamma_{\epsilon,\min}=\gamma_{\epsilon}$. Taking limits in the inequality (\ref{eq:Inq})
\begin{equation*}
\limsup_{M\to\infty} g_{\epsilon,M}(X;Y) \leq \limsup_{M\to\infty} I(Y;Z_{\gamma_{\epsilon,\min}^M}) = I(Y;Z_{\gamma_{\epsilon,\min}}).
\end{equation*}
Plugging $\gamma_{\epsilon,\min}=\gamma_{\epsilon}$ in above we conclude that
\begin{equation*}
\limsup_{M\to\infty} g_{\epsilon,M}(X;Y) \leq I(Y;Z_{\gamma_\epsilon}) = g_\epsilon(X;Y)
\end{equation*}
and therefore $\displaystyle \lim_{M\to\infty} g_{\epsilon,M}(X;Y) = g_\epsilon(X;Y)$.
\end{proof}

\end{document}